\documentclass{sig-alternate-10pt}
\usepackage{comment,pifont,array,amsmath,multirow,latexsym,tabularx}
\usepackage{graphicx}
\usepackage{balance}  
\usepackage{color}
\usepackage[boxed]{algorithm}
\usepackage{algorithmicx,algpseudocode}
\usepackage[tight]{subfigure}
\usepackage{textcomp}
\usepackage[normalem]{ulem}
\usepackage{subfigure}
\usepackage{enumerate}
\usepackage{paralist}

\usepackage[belowskip=-1ex,aboveskip=2ex,font={bf}]{caption}
\usepackage[hyphens]{url}
\usepackage[colorlinks=true,citecolor = blue, urlcolor = blue, filecolor = blue]{hyperref}
\usepackage{breakurl}

\usepackage{epsfig,graphicx,amsmath,graphics,eso-pic}

\usepackage{float}
\usepackage{array}
\usepackage{longtable}
\usepackage{amsmath}
\floatstyle{ruled}
\newfloat{program}{thp}{lop}
\floatname{program}{Program}

\usepackage{ttquot}
\usepackage{hyperref}
\usepackage{times}
\usepackage{url}

\newcommand{\eg} {{\em e.g., }}
\newcommand{\ie} {{\em i.e., }}

\newcommand{\system} {{\small\sf{BlinkDB}}}
\newcommand{\systemheader} {{BlinkDB}}
\newcommand{\systeminitalics} {{\it BlinkDB}}

{\begin{itemize}%
\setlength{\itemsep}{0in}%
\setlength{\topsep}{-.1in}%
\setlength{\partopsep}{-.1in}%
\setlength{\parskip}{0in}}%
{\end{itemize}}

{\begin{enumerate}%
\setlength{\itemsep}{0in}%
\setlength{\topsep}{-.1in}%
\setlength{\partopsep}{-.1in}%
\setlength{\parskip}{0in}}%
{\end{enumerate}}

\newcommand{\xref}[1]{\S\ref{#1}}

\usepackage[usenames,dvipsnames]{xcolor}

\newcommand{\eat}[1]{}
\newcommand{\allnotes}[1]{\textit{#1}}
\eat{
\newcommand{\ion}[1]{\allnotes{\textcolor{blue}{[Ion: #1]}}}
\newcommand{\srm}[1]{\allnotes{\textcolor{red}{[Sam: #1]}}}
\newcommand{\barzan}[1]{\allnotes{\textcolor{green}{[Barzan: #1]}}}
\newcommand{\notepanda}[1]{\allnotes{\textcolor{cyan}{[Panda: #1]}}}
\newcommand{\notesameer}[1]{\allnotes{\textcolor{gray}{[Sameer: #1]}}}
}
\newcommand{\ion}[1]{}
\newcommand{\srm}[1]{}
\newcommand{\barzan}[1]{}
\newcommand{\notepanda}[1]{}
\newcommand{\notesameer}[1]{}

\newenvironment{itemize*}%
  {\begin{itemize}%
    \setlength{\itemsep}{0pt}%
    \setlength{\parskip}{0pt}}%
  {\end{itemize}}

\newenvironment{enumerate*}%
  {\begin{enumerate}%
    \setlength{\itemsep}{0pt}%
    \setlength{\parskip}{0pt}}%
  {\end{enumerate}}

\newtheorem{theorem}{Theorem}[section]
\newtheorem{lemma}[theorem]{Lemma}

\begin{document}
%

\sloppypar


\title{\systemheader: Queries with Bounded Errors and Bounded Response Times on Very Large Data}



%
%
%
%

\numberofauthors{1} 

\author{
Sameer Agarwal$^+$,~ Aurojit Panda$^+$,~ Barzan Mozafari$^*$,~ Samuel Madden$^*$,~ Ion Stoica$^+$ \\
\\
$^+$UC Berkeley ~~ ~~~~~ ~~ $^*$MIT CSAIL\\
}

%
%


\maketitle
\begin{abstract}

  In this paper, we present \system, a massively parallel,
  sampling-based approximate query engine for running ad-hoc,
  interactive SQL queries on large volumes of data. The key insight that
  \system{} builds on is that one can often make reasonable decisions
  in the absence of perfect answers. For example, reliably detecting a
  malfunctioning server using a distributed collection of system logs
  does not require analyzing every request processed by the
  system. Based on this insight,~\system~allows one to trade-off query
  accuracy for response time, enabling interactive queries over
  massive data by running queries on data samples and presenting
  results annotated with meaningful error bars.  To achieve
  this,~\system~uses two key ideas that differentiate it from previous
  work in this area: $(1)$ an adaptive optimization framework that
  builds and maintains a set of multi-dimensional, multi-resolution
  samples from original data over time, and $(2)$ a dynamic sample
  selection strategy that selects an appropriately sized sample based
  on a query's accuracy and/or response time
  requirements.\eat{\system~supports ad-hoc queries over massive
    amounts of data and makes no assumptions about the underlying data
    distribution.}  We have built an open-source version of {\system}
  and validated its effectiveness using the well-known TPC-H benchmark
  as well as a real-world analytic workload derived from Conviva
  Inc~\cite{Conviva}. Our experiments on a $100$ node cluster show
  that \system{} can answer a wide range of queries from a real-world
  query trace on up to $17$ TBs of data in less than $2$ seconds (over
  $100\times$ faster than Hive), within an error of $2-10\%$.
  \end{abstract}

\section{Introduction}\label{introduction}

Modern data analytics applications involve computing
aggregates over a large number of records to ``roll up''
web clicks, online transactions, content downloads, phone calls,
and other features along a variety of different dimensions,
including demographics, content type, region, and so on.
Traditionally, such queries have been answered via sequential scans of
large fractions of a database to compute the appropriate statistics.
Increasingly, however, these applications demand near real-time response
rates.  Examples include (i) in a search engine, recomputing what ad(s) on websites to show to particular
classes of users as content and product popularity changes on a daily or hourly basis (e.g., based on trends on social networks like Twitter or real time search histories)
(ii) in a financial trading firm, quickly comparing the prices of securities to fine-grained historical averages to 
determine items that are under or over valued, or (iii) in a web service, determining the subset of users who 
are affected by an outage or are experiencing poor
quality of service based on the service provider or region.

\if{0}
Modern web services generate huge amounts of data that
ultimately derives its value from being analyzed in a timely
manner. Companies and their analysts often explore this
data to improve their products, increase their retention rate
and customers engagement, and diagnose problems in their
service.  Examples of such exploratory queries include:

\vspace{.05in}
\noindent{\textbf{Root cause analysis and problem diagnosis:}} 
Imagine a subset of users of a video site (\eg Netflix) experience
quality issues, such as a high level of buffering or long start-up
times. Diagnosing such problems needs to be done quickly to avoid lost
revenue. The causes can be varied: an ISP or edge cache in a certain
geographic region may be overloaded, a new OS or player release may be
buggy, or a particular piece of content may be corrupt. Diagnosing
such problems requires segmenting data across tens of dimensions to
find the particular attributes (\eg client OS, browser, firmware,
device, geolocation, ISP, CDN, content) that best characterize the
users experiencing the problem. 

\vspace{.05in} 
\noindent {\textbf{Advertising and Marketing:}} Consider a 
business that wants to adapt its policies and decisions in near real
time to maximize its revenue. This might again involve aggregating
across multiple dimensions to understand how an ad performs given a
particular group of users, content, site, and time of the
day. Performing such analysis quickly is essential, especially when
there is a change in the environment, \eg new ads, new content, new
page layout. Indeed, it can lead to a material difference if one is
able to re-optimize the ad placement every minute, rather than every
day or week.

\vspace{.05in} 
\noindent {\textbf{A/B testing:}} Consider an online service company that
  aims to optimize its business by improving the retention of it
  users.
  Often this is done by using A/B testing to experiment with anything
  from new products to slight changes in the web page layout, format,
  or colors.\eat{\footnote{A much publicized example is Google's
      testing with $41$ shades of blue for their home
      page~\cite{nytimes-ab-testing}.}} The number of combinations and
  changes that one can test is daunting, even for a company as large
  Google.  Furthermore, such tests need to be conducted carefully as
  they may negatively impact the user experience. 
  Again, in this usage scenario, it is more important to identify the
  trend fast, rather than accurately characterize the impact on every
  user.


\vspace{.1in}
\fi

In these and many other analytic applications, queries are unpredictable (because the exact problem, or query is not known in advance)
and  quick response time
is essential as data is changing quickly, and the potential profit (or loss in profit in the case of service outages)
is proportional to response time. 
Unfortunately, the conventional way of answering such queries requires scanning the entirety of several terabytes of data.
This can be quite inefficient.  For example, computing a simple average
over $10$ terabytes of data stored on $100$ machines can take in the
order of $30-45$ minutes on Hadoop if the data is striped on disks,
and up to $5-10$ minutes even if the entire data is cached in
memory. This is unacceptable for rapid problem diagnosis, and
frustrating even for exploratory analysis.  As a result, users often
employ ad-hoc heuristics to obtain faster response times, such as
selecting small slices of data (\eg an hour) or arbitrarily sampling
the data~\cite{minitables, qubole}. These efforts suggest that, at least in
many analytic applications, users are willing to forgo accuracy for achieving
better response times.

  In this paper, we introduce \system, a new distributed parallel
  approximate query-processing framework than runs on Hive/Hadoop~\cite{hive} as well as Shark~\cite{shark} (i.e., ``Hive on Spark~\cite{spark}'', which supports
caching inputs and intermediate data).  \system{}
  allows users to pose SQL-based aggregation queries over stored
  data, along with response time or error bound constraints.  Queries
  over multiple terabytes of data can be answered in seconds,
  accompanied by meaningful error bounds relative to the answer
that would be obtained if the query ran on the full data.
  The basic approach taken by \system~ is to precompute and maintain
  a carefully chosen set of random samples of the user's data,
  and then select the best sample(s) at runtime, for answering the
  query while providing error bounds using
  statistical sampling theory.  
  
While uniform samples provide a reasonable approximation for uniformly or near-uniformly distributed data, they work poorly for skewed distributions (\eg exponential or zipfian).
In particular, estimators over infrequent subgroups (\eg smartphone users in Berkeley, CA compared to New York, NY) converge slower when using uniform samples, since a larger fraction of the data needs to be scanned to produce high-confidence approximations. Furthermore, uniform samples may not contain instances of certain subgroups, leading to missing rows in the final output of queries. 
Instead, {\it stratified} or {\it biased} samples~\cite{sampling-book}, which over-represent the frequency of rare values in a sample,
 better 
represent rare subgroups in such skewed datasets.   
Therefore \system~maintains both
  a set of uniform samples, and a set of
  stratified  samples over different combinations
  of attributes. As a result,  when querying rare subgroups,
   \system~ (i) provides faster-converging estimates (\ie tighter {\it approximation errors}), and thus lower processing times, 
   compared to uniform samples~\cite{online-agg-mr}  and,  (ii) significantly reduces the number of
   missing subgroups in the query results (\ie {\it subset error}~\cite{subset-error}), enabling a wider range of applications (\eg more complex joins that 
   would be not be possible otherwise~\cite{Chaudhuri:1999}). 

  However, maintaining stratified samples over all combinations of
  attributes is impractical. Conversely, only computing stratified
  samples on columns used in past queries limits the ability 
  to handle new ad-hoc queries.  Therefore, we formulate
  the problem of sample creation as an optimization problem.  Given a
  collection of past query {\it templates} (query templates contain the set
  of columns appearing in {\tt WHERE} and {\tt GROUP BY} clauses
  without specific values for constants) and their historical
  frequencies, we choose a collection of stratified samples with total
  storage costs below some user configurable storage threshold.  These
  samples are designed to efficiently answer any instantiation of
  past query templates, and to provide good coverage for future
  queries and unseen query templates. In this paper,
  we refer to these stratified samples, constructed over different sets
  of columns (dimensions), as _multi-dimensional_ samples.
  
  In addition to multi-dimensional samples, \system~ also maintains
  _multi-resolution_ samples. For each multi-dimensional sample,
  we maintain several samples of progressively larger sizes (that we
  call multi-resolution samples). Given a query,~\system~picks
  the best sample to use at runtime.  Having samples of different
  sizes allows us to efficiently answer queries of varying complexity
  with different accuracy (or time) bounds, while minimizing the
  response-time (or error).  A single sample, would hinder our ability
  to provide as fine-grained a trade-off between speed and
  accuracy. Finally, when the data distribution or the query load
  changes, \system~ refines the solution while minimizing
  the number of old samples that need to be discarded or new samples
  that need to be generated.

Our approach is substantially different from related sampling-based
approximate query answering systems.  One line of related work is {\it
Online Aggregation}~\cite{control, online-agg, online-agg-joins} (OLA) and its 
extensions~\cite{online-agg-mr, db-online, ola-mr-pansare}.  Unlike OLA,
pre-computation and maintenance of samples allows \system~ to store each
sample on disk or memory in a way that facilitates efficient query processing
(\eg clustered by primary key and/or other attributes),
whereas online aggregation has to access the data in random order to
provide its statistical error guarantees. Additionally, unlike OLA, \system{} has
prior knowledge of the sample size(s) on which the query runs (based on 
its response time or accuracy requirements). This additional information
both helps us better assign cluster resources (\ie degree of parallelism and
input disk/memory locality), and
better leverage a number of standard distributed query optimization
techniques~\cite{join-comparison-in-mr}. There is some related work that
proposes pre-computing (sometimes stratified) samples of input data
based on past query workload characteristics~\cite{surajit-optimized-stratified,
babcock-dynamic,sciborq}. As noted above, \system~computes both
multi-dimensional (\ie samples over multiple attributes) and multi-resolution
sampling (\ie samples at different granularities), which no prior system
does. 
  We discuss related work in more detail in~\xref{related}.



In summary, we make the following contributions:
\begin{itemize*}\vspace{-.1in}
\item We develop a multi-dimensional, multi-granular stratified
  sampling strategy that provides faster convergence, minimizes missing results in the output (\ie subset error), and
 provides error/latency guarantees for ad-hoc workloads. (\xref{sec:stratified-samples}, \xref{sec:convergence-exp}) 

\item We
  cast the decision of what stratified samples to build as an optimization
  problem that takes into account: (i) the skew of the data
  distribution, (ii) query templates, and (iii) the storage
  overhead of each sample. (\xref{sec:optimal-view-creation}, \xref{sec:multi-dimension-sampling-exp})

\item We develop a run-time dynamic sample selection strategy that
  uses multiple smaller samples to quickly estimate query selectivity
  and choose the best samples for satisfying the response time and
  error guarantees. (\xref{sec:select-samplefamily}, \xref{sec:time-accuracy})
\vspace{-.1in}
\end{itemize*}

\system{} is a massively
  parallel query processing system that incorporates these ideas.  
We validate the effectiveness of
\system's design and implementation on a $100$ node cluster, using
both the TPC-H benchmarks and a real-world workload derived from
Conviva Inc~\cite{Conviva}. Our experiments show that \system~
can answer a range of queries within $2$ seconds on $17$ TB of data
within 90-98\% accuracy. Our results show that our multi-dimensional sampling approach, versus
just using single dimensional samples (as was done in previous work) can improve
query response times by up to three orders of magnitude and are further a factor of $2\times$ better than approaches
that apply online sampling at query time.
Finally, \system~is open
source\footnote{\url{http://blinkdb.org}} and several
on-line service companies have expressed interest in using it.

Next, we describe the architecture and the major components of \system.

\section{System Overview}
\label{sec:overview}


As it is built on top of Hive~\cite{hive}, {\system} supports a hybrid
programming model that allows users to write  SQL-style declarative queries
with custom user defined functions (UDFs).  In addition, for aggregation 
queries (\ie \texttt{AVG, SUM, PERCENTILE}  etc.), users can annotate
queries with either a maximum error or maximum execution time constraint.
Based on these constraints, {\system} selects an appropriately sized data
sample at runtime on which the query operates (see ~\xref{sec:example}
below for an example). Specifically, to specify an error bound, the user 
supplies a bound of the form $(\epsilon,C)$, indicating that the query should
return an answer that is within $\pm\epsilon$ of the true answer with a confidence
$C$. As an example, suppose we have a table {\it Sessions}, storing
the sessions of users browsing a media website with five columns:
\emph{Session}, \emph{Genre}, \emph{OS} (running on the user's
device), \emph{City}, and \emph{URL} (of the site visited by the
user). Then the query:

{\small
\begin{verbatim}
SELECT COUNT(*)
FROM Sessions
WHERE Genre = `western'
GROUP BY OS
ERROR WITHIN 10% AT CONFIDENCE 95%
\end{verbatim}
}
\noindent
will return the number of sessions looking at media from the
``western'' \emph{Genre} for each OS to within a relative error of
$\pm 10\%$ within a $95\%$ confidence interval. Users can also specify absolute
 errors. 
 Alternatively, users can instead request a time bound. 
For example, the query:

{\small
\begin{verbatim}
SELECT COUNT(*), RELATIVE ERROR AT 95% CONFIDENCE
FROM Sessions
WHERE Genre = `western'
GROUP BY OS
WITHIN 5 SECONDS
\end{verbatim}
}
\noindent
will return with the most accurate answer within $5$ seconds, and 
will report the estimated count along with an estimate of the relative
 error at $95\%$ confidence. This enables a user to perform
rapid exploratory analysis on massive amounts of data, wherein she can
progressively tweak the query bounds until the desired accuracy is achieved.

\subsection{Settings and Assumptions}
\label{sec:assumptions}
In this section, we discuss several assumptions we made in designing \system{}.

\noindent
\textbf{Queries with Joins.} Currently, \system~ supports two types of joins. (i) Arbitrary joins
\footnote{Note that this does not contradict the  theoretical results on the futility of uniform
  sampling for join queries~\cite{Chaudhuri:1999}, since \system~employs stratified samples for joins.}
are 
allowed (self-joins or joining two tables) 
as long as there is a stratified sample on one of the join tables that 
contains the join key in its column-set\footnote{In general, \system{} supports arbitrary $k$-way joins (i.e. joins between $k$ tables)  as long as there are at least $k-1$ stratified samples (each corresponding to one of the join operands) that all
contain the join key in their column-set.}.
(ii) In the absence of any suitable stratified sample, the join is still allowed as long as one of the two tables 
fits in memory (since \system~ does not sample tables that fit in memory). 
The latter is, however, more common in practice as data warehouses typically consist of one large de-normalized
``fact'' table (\eg ad impressions, click streams, pages served) that may need 
to be joined with other ``dimension'' tables using
foreign-keys. Dimension tables (\eg representing customers, media, or locations) 
are often small enough to fit in the aggregate memory of cluster nodes.


\noindent
\textbf{Workload Characteristics.} Since our workload is targeted at
ad-hoc queries, rather than assuming that exact queries are known a priori, 
we assume that the \emph{query templates} (\ie the set of
columns used in {\tt WHERE} and {\tt GROUP-BY} clauses) remain fairly
stable over time. We make use of this assumption when choosing which
samples to create.  This assumption has been empirically observed in a
variety of real-world production workloads~\cite{rope,
  recurring-scope} and is also true of the query trace we use for our
primary evaluation (a $2$-year query trace from Conviva Inc). We
however do not assume any prior knowledge of the specific values or
predicates used in these clauses. 
Note that, although {\system}
creates a set of stratified samples based on past query templates, at
runtime, it can still use the set of available samples to answer
any query, even if it is not from one of the historical templates.
In Section~\ref{sec:optimal-view-creation}, we show that 
our optimization framework takes into account the distribution skew of the underlying data in addition
to templates,
allowing it to perform well even when presented with previously unseen templates. 

\noindent
\textbf{Closed-Form Aggregates.}
In this paper, we focus on a small set of aggregation operators: {\tt COUNT},
{\tt SUM}, {\tt MEAN}, {\tt MEDIAN/QUANTILE}.  We estimate error of these functions using
standard estimates of closed-form error (see Table~\ref{tab:closedform}). 
However, using techniques proposed in~\cite{kai-paper}, closed-form estimates can be
easily derived for 
any combination of these basic
aggregates as well as any algebraic function that is \emph{mean-like and
asymptotically normal} (see~\cite{kai-paper} for formal definitions).

\noindent
\textbf{Offline Sampling.} 
{\system} computes samples of input data and reuses them across many
queries.  One challenge with any system like \system{} based on
offline sampling is that there is a small but non-zero probability
that a given sample may be non-representative of the true data, e.g.,
that it will substantially over- or under-represent the frequency of
some value in an attribute compared to the actual distribution, such
that a particular query $Q$ may not satisfy the user-specified error
target.  Furthermore, because we do not generate new samples for each
query, no matter how many times $Q$ is asked, the error target will
not be met, meaning the system can fail to meet user specified
confidence bounds for $Q$.  Were we to generate a new sample for every
query (i.e., perform {\it online} sampling), our confidence bounds
would hold, because our error estimates ensure that the probability of
such non-representative events is proportional to the user-specified
confidence bound.  Unfortunately, such re-sampling is expensive and
would significantly impact query latency.  Instead, our solution is to
periodically replace samples with new ones in the background, as
described in~\xref{sec:sample-maintenance}.

\subsection{Architecture}
\label{sec:arch}

Fig.~\ref{arch} shows the overall architecture of \system.
{\system} builds on the Apache Hive framework~\cite{hive} and
adds two major components to it: (1) an offline sampling module that creates and
maintains samples over time, and (2) a run-time sample selection module that creates
an {\it Error-Latency Profile (ELP)} for ad-hoc queries. The ELP characterizes
the rate at which the error (or response time) decreases (or increases) as
 the size of the sample on which the query operates increases. This is used to
select a sample that best
satisfies the user's constraints. {\system} augments the query parser,
optimizer, and a number of aggregation operators to allow queries to specify
constraints for accuracy, or execution time.

\begin{figure}[htbp]
\begin{center}
\includegraphics*[width=2.5in]{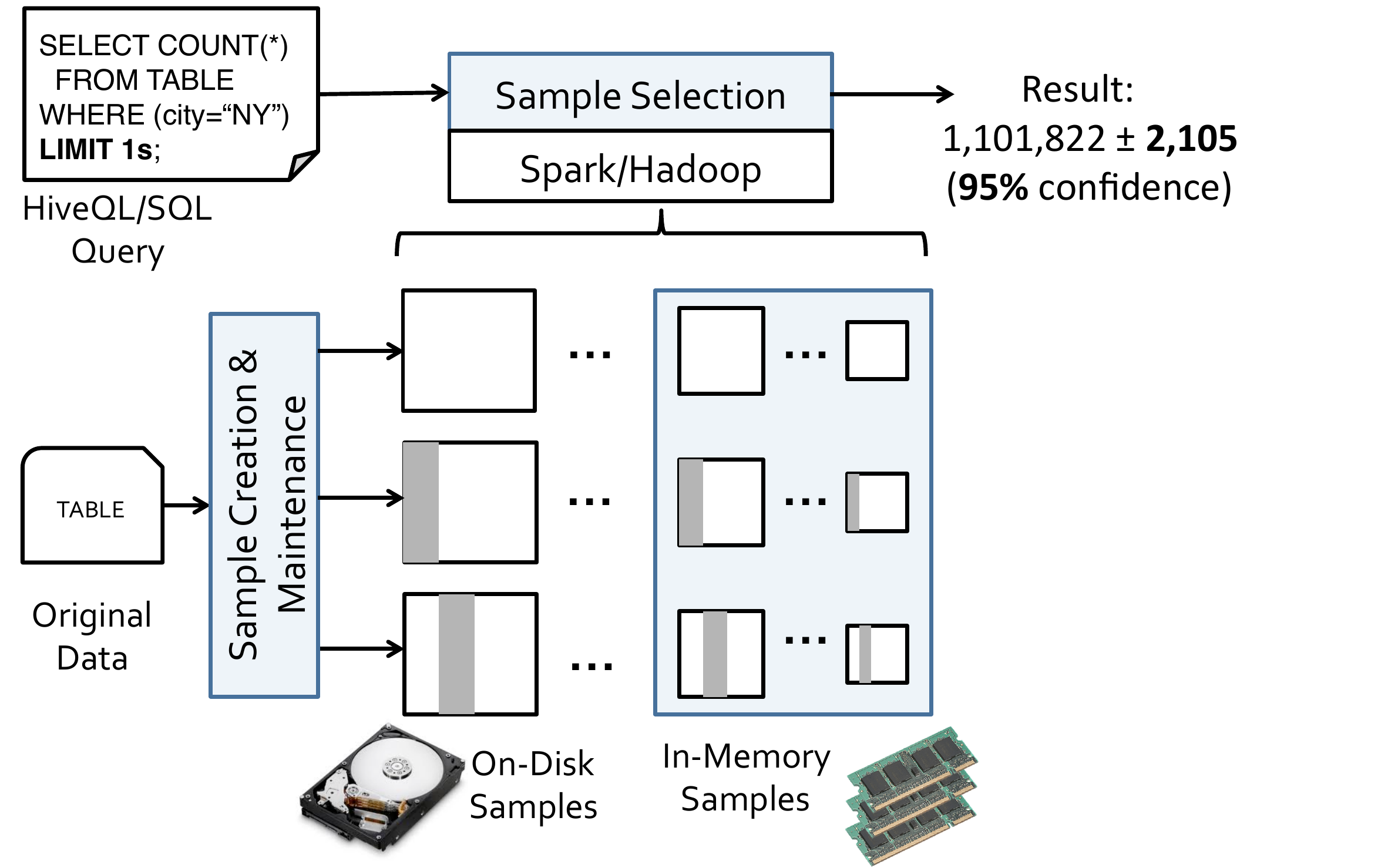}
\vspace{-.15in}
\caption{{\system} architecture.}
\label{arch}
\end{center}
\vspace{-.2in}
\end{figure}


\subsubsection{Offline Sample Creation and Maintenance}
\label{sec:sample-creation}

This component is responsible for creating and maintaining a set of
uniform and stratified samples. We use uniform samples over the entire
dataset to handle queries on groups of columns with relatively uniform
distributions, and stratified samples (on one or more columns) to
handle queries on groups of columns with less uniform distributions. This
component consists of three sub-components:

\begin{asparaenum}

\item \textbf{Offline Sample Creation.}  Based on statistics collected
  from the data (\eg average row sizes, key skews, column histograms etc.), and
  historic query templates, {\system} computes a set of uniform samples
  and multiple sets of stratified samples from the underlying data.
  We rely on an optimization framework described in \xref{sec:optimal-view-creation}. 
  Intuitively, the optimization framework builds
  stratified samples over column(s) that are (a) most useful for the
  query templates in the workload, and (b) most skewed, \ie they
  have long-tailed distributions where rare values are more likely to
  be excluded by a uniform sample.
  
\item \textbf{Sample Maintenance.} As new data arrives, we
  periodically update the initial set of samples. Our update strategy
  is designed to minimize performance overhead and
  avoid service interruption.  A monitoring module observes
  overall system performance, detecting any significant changes
  in data distribution (or workload), and triggers periodic
  sample replacement, and updates, deletions, or creations of new
  samples.

\item \textbf{Storage optimization.} In addition to caching samples in
memory, to maximize disk throughput, we
  partition each sample into many small files, and leverage the
  block distribution strategy of HDFS~\cite{hdfs} to spread those files across
  the nodes in a cluster. Additionally, we optimize the storage overhead,
  by recursively building larger samples as a union of smaller samples that are
  built on the same set of columns.



\end{asparaenum}


\subsubsection{Run-time Sample Selection} 
\label{sec:sample-slection}

Given a query, we select an optimal sample at runtime so as to
meet its accuracy or response time constraints.
We do this by dynamically running
the query on smaller samples to estimate the query's
selectivity, error rate, and response time, and then extrapolate
to a sample size that will satisfy user-specified error or response time 
goals. 
\xref{solution:selection} describes this procedure in
detail.

\subsection{An Example}
\label{sec:example}

To illustrate how {\system}~operates, consider a table derived from a log
of downloads by users from a media website,
as shown in Figure~\ref{fig:startified-samples-example}. The
table consists of five columns: \emph{Session}, \emph{Genre},
\emph{OS}, \emph{City}, and \emph{URL}.

Assume we know the query templates in the workload, and that
$30\%$ of the queries had \emph{City} in their {\tt WHERE/GROUP BY}
clause, $25\%$ of the queries had \emph{Genre}~{\tt AND}~\emph{City}
in their {\tt WHERE/GROUP BY} clause, and so on.

\begin{figure}[htbp]
\begin{center}
\includegraphics*[width=225pt]{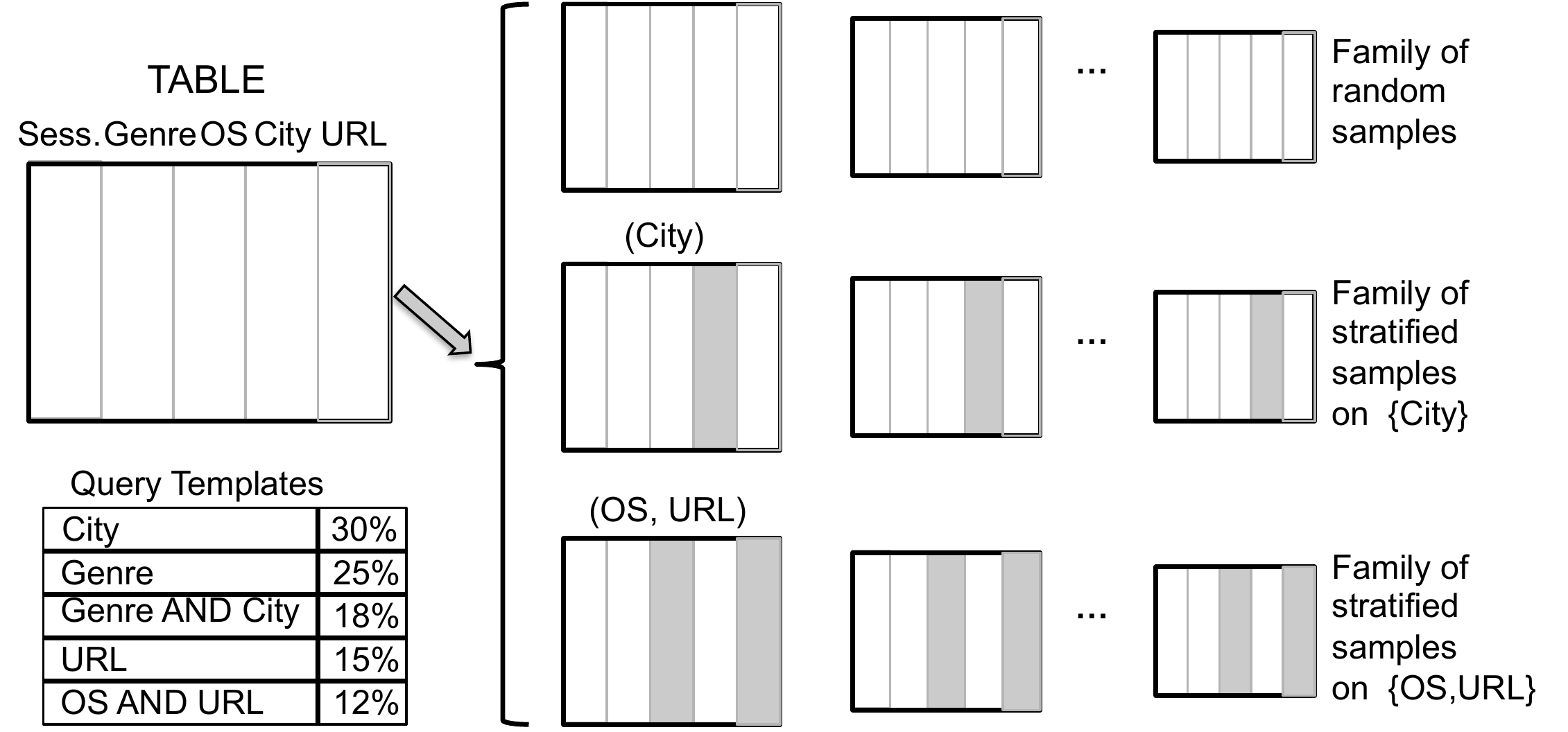}
\caption{An example showing the samples for a table with five columns,
  and a given query workload.}
\label{fig:startified-samples-example}
\end{center}
\vspace{-.2in}
\end{figure}

Given a storage budget, {\system} creates several multi-dimen\-sion\-al
and multi-resolution samples based on past query templates and
the data distribution. These samples are organized in \emph{sample
families}, where each family contains multiple samples of different
granularities. One family consists of uniform samples, while the other
families consist of stratified samples biased on a given set of
columns. In our example, {\system} decides to create two sample
families of stratified samples: one on \emph{City}, and another one on
$(\emph{OS}, \emph{URL})$. Note that despite \emph{Genre} being a 
frequently queried column, we
do not create a stratified sample on this column.  This could be due
to storage constraint or because \emph{Genre} is uniformly
distributed, such that queries that only use this column are
already well served by the uniform sample. Similarly, {\system} does
not create stratified samples on columns
$(\emph{Genre},~\emph{City})$, in this case because queries on these
columns are well served by the stratified samples on the \emph{City}
column. {\system} also creates several instances of each sample
family, each with a different size, or {\it resolution}.  For
instance, {\system} may build three biased samples on
columns$(\emph{OS},~\emph{URL})$ with $1$M, $2$M, and $4$M tuples
respectively.  In \xref{sec:optimal-view-creation}, we present
an algorithm for optimally picking these sample families.

For every query, at run time, {\system} selects the appropriate sample
family and the appropriate sample resolution to answer the query based on the user
specified error or response time bounds.  In
general, the columns in the {\tt WHERE/GROUP BY} of a query may not
exactly match any of the existing stratified samples. For example,
consider a query, $Q$, whose {\tt WHERE} clause is
(\emph{OS}$=$\emph{'Win7'} {\tt AND} \emph{City}$=$\emph{'NY'} {\tt
  AND} \emph{URL}$=$'\url{www.cnn.com}'). In this case, it is not
clear which sample family to use. To get around this problem,
{\system} runs $Q$ on the smallest resolutions of other candidate
sample families, and uses these results to select the appropriate
sample, as described in detail in ~\xref{solution:selection}.

\section{Sample Creation}\label{solution:sampling}

As described in~\xref{sec:sample-creation}, \system~creates a
set of multi-dimensional, multi-resolution samples to accurately and
quickly answer ad-hoc queries. In this section, we describe sample creation in
detail.  First, in~\xref{sec:stratified-samples}, we
discuss the creation of a sample family, a set
of stratified samples of different sizes, on the same set of
columns. In particular, we show how the choice of stratified samples
impact the query's accuracy and response time, and evaluate the
overhead for skewed distributions. Next, in~\xref{sec:optimal-view-creation} we formulate and solve an
optimization problem to decide on the sets of columns on which we build sample
families.  

\subsection{Multi-resolution Stratified Samples}
\label{sec:stratified-samples}

In this section, we describe our techniques for constructing a family
of stratified samples from input tables.  We describe how we maintain
these samples in~\xref{sec:sample-maintenance}. Table~\ref{tab:notations}
contains the notation used in the rest of this section.

\begin{table}
\begin{center}
{\small
  \begin{tabular}{| l | l |}
    \hline
    {\bf Notation} & {\bf Description} \\ \hline\hline
    $T$ & fact (original) table \\ \hline
    $\phi$ & set of columns in $T$ \\ \hline
    $R(p)$ & random sample of $T$, where each row in $T$ \\
              & is selected with probability $p$ \\ \hline
   $S(\phi, K)$ & stratified sample associated to $\phi$, where \\
                       & frequency of every value $x$ in $\phi$ is capped by $K$ \\ \hline
    $SFam(\phi)$ & family (sequence) of multi-dimensional multi-\\
                          & resolution stratified samples associated with $\phi$ \\ \hline
    $F(\phi, S, x)$ & frequency of value $x$ in set of columns $\phi$ in \\
                           & sample/table $S$\\ \hline
  \end{tabular}
}
\end{center}
\vspace{-.1in}
\caption{Notation in \xref{sec:stratified-samples}}
\vspace{-.1in}
\label{tab:notations}
\end{table}

Queries on uniform samples  converge quickly to the true answer, when the original data is distributed uniformly
or near-uniformly. This convergence, is however, much slower for uniform samples over 
highly skewed distributions (\eg exponential or Zipfian) because a much larger fraction of the entire data set needs to be scanned to produce high-confidence estimates on infrequent values.
A second, perhaps more important, problem is that uniform samples may not contain any instances of certain subgroups, leading to missing rows in the final output of queries.
The
standard approach for dealing with such distributions is to use
stratified sampling~\cite{sampling-book}, which ensures that rare subgroups are {\it sufficiently} represented
in such skewed datasets.  This both provides faster convergence of answer estimates and avoids missing subgroups in results. In this
section, we describe the use of stratified sampling in \system. 

Let $\phi = \{c_1, c_2, \ldots, c_k\}$ be a subset of columns in the
original table, $T$. For any such subset we define a \emph{sample family}
as a sequence of stratified samples over $\phi$
(see Table~\ref{tab:notations} for our notations):

\vskip -0.1in
\begin{equation}
\label{eq:sample-family-def}
SFam(\phi) = \{S(\phi, K_i) \mid 0 \leq i < m \},
\end{equation}

\noindent
where $m$ is the number of samples in the family.  By maintaining
multiple stratified samples for the same column subset $\phi$ we allow
a finer granularity tradeoff between query accuracy and response
time. In the remainder of this paper, we use the term ``set'', instead
of ``subset'' (of columns), for brevity.
In~\xref{sec:optimal-view-creation}, we describe how to select
the sets of columns on which sample families are built.

\begin{figure}[htbp]
\begin{center}
\vspace{-.1in}
\includegraphics*[width=225pt]{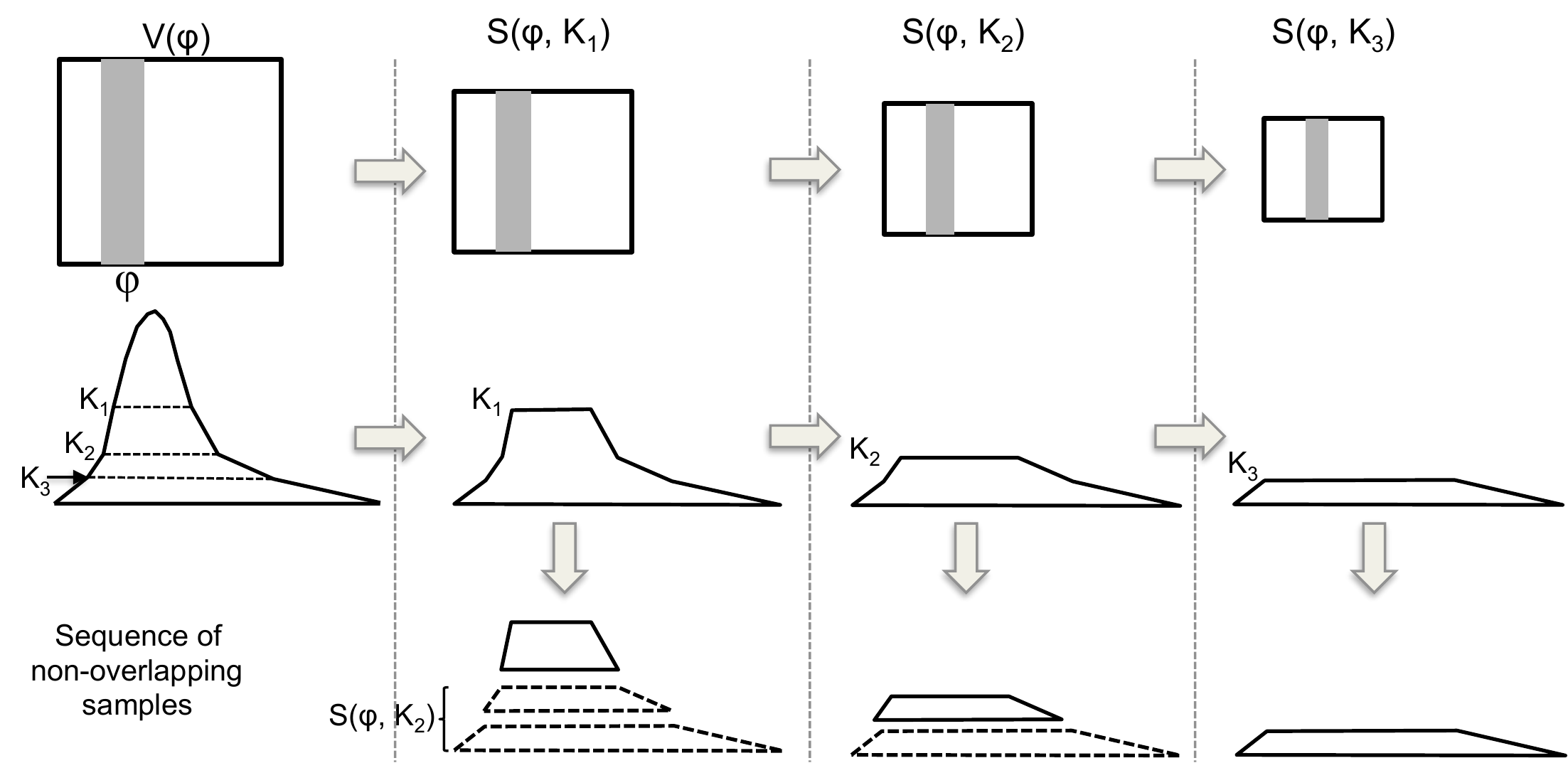}
\caption{Example of stratified samples associated with a set
  of columns, $\phi$.}
\label{fig:stratified-sample}
\end{center}
\vspace{-.1in}
\end{figure}

A stratified sample $S(\phi, K_i)$ on the set of columns, $\phi$, caps
the frequency of every value $x$ in $\phi$ to $K_i$.\footnote{Although
  stratification is done on column set $\phi$, the sample that is
  stored contains all of the columns from the original table.} More
precisely, consider tuple $x = <x_1, x_2, \ldots, x_k>$, where $x_i$
is a value in column $c_i$, and let $F(\phi, T, x)$ be the frequency
of $x$ in column set $\phi$ in the original table, $T$. If $F(\phi, T,
x) \leq K_i$, then $S(\phi, K_i)$ contains all rows containing $x$ in
$T$. Otherwise, if $F(\phi, T, x) > K_i$, then $S(\phi, K_i)$ contains
$K_i$ randomly chosen rows from $T$ that contain $x$.

Figure~\ref{fig:stratified-sample} shows a sample family associated
with column set $\phi$. There are three stratified samples $S(\phi,
K_1)$, $S(\phi, K_2)$, and $S(\phi, K_3)$, respectively, where $K_1$
is the largest sample, and $K_3$ the smallest. Note that since each
sample is a subset of a bigger sample, in practice there is no need to
independently allocate storage for each sample. Instead, we can
construct smaller samples from the larger ones, and thus need an amount
of storage equivalent to maintaining only the largest sample. This
way, in our example we only need storage for the sample corresponding
to $K_1$, modulo the metadata required to maintain the smaller
samples.

  Each stratified sample $S(\phi, K_i)$ is stored sequentially sorted
  according to the order of columns in $\phi$. Thus, the records with the
  same or consecutive $x$ values are stored contiguously on the disk, which, as we will see,
  significantly improves the execution times or range of the queries on the set of
  columns $\phi$.

 Consider query $Q$ whose {\tt WHERE} or {\tt GROUP BY} clause
  contains $(\phi = x)$, and assume we use $S(\phi, K)$ to answer this
  query. If $F(\phi, S(\phi, K), x) < K$, the answer is exact as the
  sample contains all rows from the original table. On the other hand,
  if $F(\phi, S(\phi, K), x) > K$, we answer $Q$ based on $K$ random
  rows in the original table. For the basic aggregate operators {\tt
    AVG}, {\tt SUM}, {\tt COUNT}, and {\tt QUANTILE}, $K$ directly
  determines the error of $Q$'s result. In particular, for these
  aggregate operators, the standard deviation is inversely
  proportional to $\sqrt{K}$, as shown in Table~\ref{tab:closedform}.

  In this paper, we choose the samples in a family so that they have
  exponentially decreasing sizes.  In particular, $K_i =
  \lfloor K_1/c^{i} \rfloor$ for $(1 \leq i \leq m)$, and $m = \lfloor
  {\log}_c K_1 \rfloor$. Thus, the cap of samples in the sequence
  decreases by factor $c$.

\noindent
\textbf{Properties.} A natural question is how ``good'' is a sample
family, $SFam(\phi)$, given a specific query, $Q$, that executes on
column set $\phi$. In particular, let $S(\phi, K^{opt})$ be the
smallest possible stratified sample on $\phi$ that satisfies the error
or response time constraints of $Q$. Since $SFam(\phi)$ contains only
a finite number of samples, $S(\phi, K^{opt})$ is not guaranteed to be
among those samples. Assume $K_1 \geq K^{opt} \geq \lfloor K_1/c^m
\rfloor$, and let $S(\phi, K')$ be the closest sample in $SFam(\phi)$
that satisfies $Q$'s constraints. Then we would like that the $Q$'s
performance when running on $S(\phi, K^{opt})$ to be as close as
possible to $Q$'s performance when running on the optimal-sized
sample, $S(\phi, K^{opt})$. Then, for $K^{opt} \gg c$ the following
two properties hold (see Appendix \ref{sec:analysis} for proofs):
\vspace{-0.1in}
\begin{enumerate*}
\item For a query with response time constraints, the response time of
  the query running on $S(\phi, K')$ is within a factor of $c$ of the
  response time of the query running on the optimal-sized sample,
  $S(\phi, K^{opt})$.

\item For a query with error constraints, the standard deviation of
  the query running on $S(\phi, K')$ is within a factor of $\sqrt{c}$
  of the response time of the query running on $S(\phi, K^{opt})$.
\end{enumerate*}
\vspace{-0.1in}

\noindent
\textbf{Storage overhead.} Another consideration is the overhead associated with
maintaining these samples, especially for heavy-tailed distributions. In
Appendix \ref{sec:analysis} we provide numerical results for a Zipf
distribution, one of the most common heavy-tailed
distributions. Consider a table with $1$ billion tuples and a column
set with a Zipf distribution with an exponent of $1.5$. Then, the
storage required by a family of samples $S(\phi, K)$ is only $2.4\%$
of the original table for $K_0=10^4$, $5.2\%$ for $K_0=10^5$, and
$11.4\%$ for $K_0=10^6$.

These results are consistent with real-world data from Conviva Inc,
where for $K_0 = 10^5$, the overhead incurred for sample families on popular
columns like city, customer, autonomous system number (ASN) are all
less than $10\%$.

\subsection{Optimization Framework}
\label{sec:optimal-view-creation}

We now describe the optimization framework we developed to select subsets
of columns on which to build sample families.  Unlike prior work which
focuses on single-column stratified
samples~\cite{babcock-dynamic}, \system{} creates multi-dimensional
(\ie multi-column) stratified samples.  Having stratified samples on
multiple columns that are frequently queried together
can lead to significant improvements in
both query accuracy and latency, especially when the
set of columns have a skewed joint distribution.
However, these samples lead to an increase in the
storage overhead because (1) samples on multiple columns can be
larger than single-column samples since multiple columns
often contains more unique values than individual columns, 
and (2) there are an exponential number of subsets of columns, all of which may not fit
in our storage budget. As a result, we need to be
careful in choosing the set of columns on which to build
stratified samples.  Hence, we formulate the trade off between
storage and query accuracy/performance as an optimization problem,
described next.

\subsubsection{Problem Formulation}
\label{sec:formulation}
The optimization problem takes three factors into account in determining
the sets of columns on which stratified samples should be built: the \emph{non-uniformity/skew 
of the data}, \emph{workload
  characteristics}, and the \emph{storage cost of samples}. 

\vspace{.1in} 
\noindent 
\textbf{Non-uniformity (skew) of the data.}  Intuitively, the greater the
skew  for a set of columns, the more important it is to have a
stratified sample on those columns. If there is no skew, the uniform
sample and stratified sample will be identical.
Formally, for a subset of columns $\phi$ in table $T$, let $D(\phi)$
denote the set of all distinct values appearing in $\phi$. Recall from
Table~\ref{tab:notations} that $F(\phi, T, v)$ is the frequency of
value $v$ in $\phi$.
Let $\Delta(\phi)$ be a non-uniformity metric on the distribution of
the values in $\phi$. The higher the non-uniformity in $\phi$'s
distribution the higher the value of $\Delta(\phi)$. When $\phi$'s
distribution is uniform (\ie when $F(\phi, T,
v)=\frac{|D(\phi)|}{|T|}$ for $v\in D(\phi)$),  
$\Delta(\phi)=0$.  
In general, $\Delta$ could be any metric of the distribution's skew (e.g., \emph{kurtosis}). 
In this paper, for the sake of simplicity, we use a more intuitive notion of non-uniformity, defined as:
\begin{displaymath}
\Delta(\phi)=|\{v\in D(\phi) | F(\phi,T,v)<K\}|
\end{displaymath}
where $K$ represents the cap corresponding to the largest sample in
the family, $S(\phi, K)$ (see~\xref{sec:stratified-samples}).
 Intuitively, this metric
captures the length of $\phi$'s tail, i.e., the number of unique values
in $\phi$ whose frequencies are less than $K$. While the
 rest of this paper 
uses this metric, our framework allows other metrics to be used.

\vspace{.1in}
\noindent
\textbf{Workload.} The utility of a stratified sample increases if the
set of columns it is biased on occur together frequently in queries.  One
way to estimate such co-occurrence is to use the frequency with which columns
have appeared together in past queries. 
However, we wish to avoid over-fitting to a particular set of queries
  since future queries may use different columns. Hence,
   we  use a \emph{query workload}
  defined as a set of $m$ query
templates and their weights: 
$$\langle \phi_{1}^T,
w_{1}\rangle, \cdots, \langle \phi_{m}^T, w_{m}\rangle$$ where
$0<w_{i}\leq 1$ is the weight (normalized frequency or importance) of
the $i$'th query template and $\phi_{i}^T$ is the set of columns
appearing in the $i$'th template's {\tt WHERE} and {\tt GROUP BY}
clauses\footnote{Here, {\tt HAVING} clauses are treated as columns in
  the {\tt WHERE} clauses.}.
  
\vspace{.1in} 
\noindent 
\textbf{Storage cost.} 
Storage is the main constraint against building too many multi-dimensional
sample families, and thus, our optimization framework takes the
storage cost of different samples into account.  We use $Store(\phi)$
to denote the storage cost (say, in MB) of building a sample family on a
set of columns $\phi$.  

Given these three factors defined above, we now introduce
our optimization formulation. Let the overall storage budget be
$\mathbb{S}$. Consider the set of $\alpha$ column combinations that
are {\it candidates} for building sample families on, say
$\phi_{1},\cdots,\phi_{\alpha}$. For example, this set can include all
column combinations that co-appeared at least in one of the query templates. Our goal is to
select $\beta$ subsets among these candidates, say
$\phi_{i_1},\cdots,\phi_{i_{\beta}}$, such that
$$\sum_{k=1}^{\beta}Store(\phi_{i_k}) \leq \mathbb{S}$$ and these
subsets can ``\emph{best}'' answer our queries.

Specifically, in \system, we maximize the following mixed linear integer program (MILP):
\begin{equation}
G=\sum_{i=1}^{m} w_{i} \cdot y_{i} \cdot \Delta(\phi^{T}_{i})
\label{eq:goal}
\end{equation}
\vskip -0.1in
 subject to 
\vskip -0.1in
\begin{equation}
\sum_{j=1}^{\alpha} Store(\phi_{j}) \cdot z_{j} \leq \mathbb{S}
\label{eq:storage}
\end{equation}
\vskip -0.1in
 and 
\vskip -0.1in
\begin{equation}
\forall 1\leq i\leq m: ~~ y_{i} \leq  \underset{\phi_{j}\subseteq \phi_{i}^T }{\max} \frac{|D(\phi_{j})|}{|D(\phi_{i}^T)|} \cdot z_{j}  
\label{eq:coverage}
\end{equation}
\noindent
where $0\leq y_{i}\leq 1$ and $z_{j}\in\{0,1\}$. 

Here, $z_{j}$ variables determines whether a sample family is built or not,
i.e., when $z_{j}=1$, we build a sample family on $\phi_{j}$;
otherwise, when $z_{j}=0$, we do not.

The goal function (\ref{eq:goal}) aims to maximize the weighted sum of
the coverage of the query templates. The degree of coverage of query
template $\phi^{T}_i$ with a set of columns $\phi_j \subseteq
\phi^{T}_i$, is the probability that a given value in $\phi^{T}_i$ is
also present in the stratified sample associated with $\phi_j$, i.e.,
$S(\phi_j, K)$.  Since this probability is hard to compute in
practice, in this paper we approximate it by $y_i$ value which
is determined by constraint (\ref{eq:coverage}). The $y_i$ value is in
$[0, 1]$, with $0$ meaning no coverage, and $1$ meaning full
coverage. The intuition behind (\ref{eq:coverage}) is that when we
build a stratified sample on a subset of columns $\phi_{j}\subseteq
\phi^{T}_{i}$, namely when $z_{j}=1$, we have partially covered
$\phi^{T}_{i}$ too. We compute this coverage as the ratio of the
number of unique values between the two sets, i.e.,
$|D(\phi_{j})|/|D(\phi_{i}^T)|$. When the number of unique values in
$\phi_j$ and $\phi^T_i$ are the same we are guaranteed to see all the
unique values of $\phi^T_i$ in the stratified sample over $\phi_j$ and
therefore the coverage will be $1$.

Finally, we need to weigh the coverage of each set of columns by their
importance: a set of columns $\phi^{T}_{i}$ is more important to 
cover when (1) it has a higher frequency, which is represented by
$w_{i}$, or (2) when the joint distribution of $\phi^{T}_{i}$ is more
skewed (non-uniform), which is represented by $\Delta(\phi^{T}_{i})$.
Thus, the best solution is when we maximize the sum of $w_{i}
\cdot y_{i} \cdot \Delta(\phi_{i})$ for all query templates, as
captured by our goal function (\ref{eq:goal}).

Having presented our basic optimization formulation,
 we now address the problem of choosing the
initial candidate sets, namely $\phi_{1},\cdots,\phi_{\alpha}$.  In~\xref{sec:drift}
we discuss how this problem formulation
handles changes in the data distribution as well as changes of
workload.

\subsubsection{Scaling the Solution}
\label{sec:candidates}
Naively, one can use the power set of all the columns as the set of
candidate column-sets. However, this leads to an exponential  number of variables in the MILP formulation and thus, becomes impractical for
tables with more than $O(20)$ columns.  To reduce this exponential search space,
we restrict
the candidate subsets to only those that have appeared together at least in one of the
query templates (namely, $\{\phi|\exists~i,~\phi\subseteq\phi_{i}^T\}$).  This
does not affect the optimality of the solution, because a column $A$
that has not appeared with the rest of the columns in $\phi$ can be
safely removed without affecting any of the query templates. In our experiments, we have been able
to solve our MILP problems with $O(10^6)$ variables within $6$ seconds
using an open-source solver \cite{glpk}, on a commodity server.
However, when the {\tt WHERE/GROUP BY} clauses of the query templates
exceed $O(20)$ columns, the number of variables can exceed $10^6$.  In
such cases, to cope with this combinatorial explosion, 
we further limit candidate subsets 
to those  consisting of no more than a fixed
number of columns, say $3$ or $4$ columns. 
This too has  proven to be a safe
restriction since, in practice, subsets with a large number of columns
have many unique values, and thus, are not chosen by the optimization framework due to 
 their high storage cost.

 
\subsubsection{Handling Data/Workload Variations}
\label{sec:drift}

Since \system~is designed to handle ad-hoc queries,
our optimization formulation is designed to avoid {\it over-fitting} samples to past queries by:  
(i) only looking at the set of columns that appear in the query templates instead optimizing for specific constants in queries and (ii) considering infrequent subsets with a high degree of skew (captured by $\Delta$ in~\xref{sec:formulation}).  

In addition, \system~periodically (currently, daily) updates data and workload
statistics to decide whether the current set of sample families are
still effective or if the optimization problem needs to be re-solved based
on the new input parameters. When re-solving the optimization, 
\system{} tries to find a solution that is robust to workload changes
by favoring sample families that require fewer changes to
the existing set of samples, as described below. 
Specifically, \system~
allows the administrator
to decide what percentage of the sample families (in terms of storage
cost) can be discarded/added to the system whenever \system~triggers
the sample creation module as a result of changes in data or workload
distribution. The administrator makes this decision by manually setting a parameter $0\leq r\leq 1$, 
which is incorporated into an extra constraint in our MILP formulation:

\vskip -0.1in
 \begin{equation}
\sum_{j=1}^{\alpha} (\delta_{j}- z_{j})^{2} \cdot S (\phi_{j}) \leq r\cdot \sum_{j=1}^{\alpha} \delta_{j} \cdot S (\phi_{j})
\label{eq:drift}
 \end{equation} 
 
Here $\delta_{j}$'s are additional input parameters stating whether 
$\phi_{j}$ already exists in the system (when $\delta_{j}=1$) or it does not ($\delta_{j}=0$).
In the extreme case, when the administrator chooses $r=1$, 
the constraint (\ref{eq:drift}) will trivially hold and thus, the
 sample creation module is free to create/discard any sample
families, based on the other constraints discussed in~\xref{sec:formulation}. 
On the other hand, setting
$r=0$ will completely disable
this module in~\system, i.e., no new samples will be created/discarded
because $\delta_{j}=z_{j}$ will be enforced for all $j$'s. For values of
 $0<r<1$, we ensure that the total size of the samples that need to created/discarded is at most a fraction $r$ of the total size of existing samples in the system (note than we have to {\it create} 
a new sample when $z_{j}=1$ but $\delta_{j}=0$ and need to {\it delete} an existing sample when 
$z_{j}=0$ but $\delta_{j}=1$).
When \system~runs the optimization problem for the first time
 $r$ is always set to $1$.

 \section{BlinkDB Runtime} 
 \label{solution:selection}

 In this section, we provide an overview of query execution in \system.
and our approach for online sample
 selection.  Given a query $Q$, the goal is to select one (or more) sample(s)
 at~\emph{run-time} that meet the specified time or error
 constraints and then compute answers over them. Selecting a sample involves first selecting a _sample
 family_ (\ie~dimension), and then selecting a _sample resolution_
 within that family. The selection of a sample family depends on the
 set of columns in $Q$'s clauses, the selectivity of its selection
 predicates, and the data distribution. In turn, the selection of the
 resolution within a sample family depends on $Q$'s time/accuracy
 constraints, its computation complexity, and the physical
 distribution of data in the cluster. 

 As with traditional query processing, accurately predicting the
 selectivity is hard, especially for complex {\tt WHERE} and {\tt
   GROUP-BY} clauses. This problem is compounded by the fact that the
 underlying data distribution can change with the arrival of new data.
 Accurately estimating the query response time is even harder,
 especially when the query is executed in a distributed fashion. This
 is (in part) due to variations in machine load, network throughput,
 as well as a variety of non-deterministic (sometimes time-dependent)
 factors that can cause wide performance fluctuations.

 Rather than try to model selectivity and response time, our sample selection strategy takes
 advantage of the large variety of non-overlapping samples in
 {\system} to estimate the query error and response time at
 \emph{run-time}. In particular, upon receiving a query,
 \system~``probes'' the smaller samples of one or more sample families
 in order to gather statistics about the query's selectivity, complexity and
 the underlying distribution of its inputs.  Based on these results,
 {\system} identifies an optimal sample family and resolution to run the query on.

 In the rest of this section, we explain our query execution, by first
 discussing our mechanism for selecting a sample family
 (\xref{sec:select-samplefamily}), and a sample size (\xref{sec:select-samplesize}).
 We then discuss how to produce unbiased results from 
 stratified samples (\xref{sec:unbiased-biased}), followed by
re-using intermediate data  in \system{} (\xref{sec:reusing}).

\subsection{Selecting the Sample Family}
\label{sec:select-samplefamily}

Choosing an appropriate sample family for a query
primarily depends on the set of columns used for _filtering_ and/or _grouping_.
The {\tt WHERE} clause itself  may either
consist of conjunctive predicates ({\tt condition1 AND condition2}), disjunctive
predicates ({\tt condition1 OR condition2}) or a combination of the two.
 Based on
this, {\system} selects one or more suitable sample families for the
query as described in ~\xref{sec:conjunctive-predicates} and~\xref{sec:disjunctive-predicates}. 

\begin{table*}[t]
\vspace{-.1in}
{\small
\hfill{}
\begin{tabular}{| p{0.08\textwidth} | p{0.42\textwidth} | p{0.42\textwidth} |}
\hline
\bf Operator & \bf Calculation & \bf Variance\\
\hline
\hline
\texttt{Avg} & $\frac {\sum{X_i}}{n}$ [$X_i$: observed values; $n$: sample size) & $\frac{S_n^2}{n}$ ($S_n^2$: sample variance]
\\\hline
\texttt{Count} & $\frac{N}{n}\sum{\mathbf{I}_{K}}$ [$\mathbf{I}_{K}$: matching tuple indicator; $N$: Total Rows] & $\frac{N^2}{n} c(1-c)$ [$c$: fraction of items which meet the criterion]
\\\hline
\texttt{Sum} & $\left(\frac{N}{n}\sum{\mathbf{I}_{K}} \right) \bar{X}$ & $N^2\frac{S_n^2}{n} c(1-c)$\\\hline
\texttt{Quantile} & 
$x_{\lfloor h\rfloor}+ (h - \lfloor h\rfloor)(x_{\lceil h\rceil} -x_{\lfloor h\rfloor})$ [$x_i$: $i^{th}$ ordered element & $\frac
{1}{f(x_p)^2}\frac{p(1-p)}{n}$ [$f$: pdf for data] \\
& in sample;  $p$: specified quantile; $h$: $p\times n$] &\\\hline
\end{tabular}
\hfill{}
}
\vspace{-.1in}
\caption{Error estimation formulas for common aggregate operators.}
\vspace{-.1in}

\label{tab:closedform}
\end{table*}


\subsubsection{Queries with Conjunctive Predicates} 
\label{sec:conjunctive-predicates}
Consider a query $Q$ whose {\tt WHERE} clause
contains only conjunctive predicates. Let $\phi$ be the set of columns
that appear in these clause predicates. If $Q$ has multiple {\tt
  WHERE} and/or {\tt GROUP BY} clauses, then $\phi$ represents the
union of the columns that appear in each of these
predicates.
If \system finds one or more stratified sample family on a set of columns
$\phi_i$ such that $\phi \subseteq \phi_i$, we simply pick the $\phi_i$ with the smallest number of columns, 
and run the query on $SFam(\phi_i)$. 
However, if there is no stratified sample on a column
set that is a superset of $\phi$, we run $Q$ in parallel on the smallest
sample of all sample families currently maintained by the system. 
Then, out of these samples we select the one that corresponds to the highest ratio
of (i) the number of rows \emph{selected} by $Q$, to (ii) the number of
rows \emph{read} by $Q$ (\ie number of rows in that sample). Let $SFam(\phi_i)$ be the family containing
this sample. The intuition behind this choice is that the response
time of $Q$ increases with the number of rows it reads, while the
error decreases with the number of rows $Q$'s {\tt WHERE} clause selects.

A natural question is why probe all sample families, instead of
only those built on columns that are in $\phi$? The reason is simply
because the columns in $\phi$ that are missing from a family's column
set, $\phi_i$, can be negatively correlated with the columns in
$\phi_i$. 
In addition, we expect the smallest sample of each family to
fit in the aggregate memory of the cluster, and thus running $Q$ on
these samples is very fast. 
\eat{
Let $n_{i, m}$ be the number of rows selected by $Q_i$ when running on
the smallest sample of the selected family, $S(\phi_i, K_m)$. If
$n_{i,m} < K_m$, then we are done as $S(\phi_i, K_m)$ contains all
rows in the original table that match $Q$, and thus, in this case, we
get an exact answer.  Otherwise, we select sample $S(\phi_i, K_q)$
where $K_q$ is the smallest value in $SFam(\phi)$ that is larger than
$K_m n/n_{i,m}$. This ensures that the expected number of rows
selected by $Q$ when running on sample $S(\phi_i, K_q)$ is $\geq
n$. As a result, the answer of $Q$ on $S(\phi_i, K_q)$ will meet
$Q_i$'s error constraint.

\vspace{.1in} 
\noindent{\textbf{Response time constraints:}} If $Q$ specifies a
response time constraint, we select the sample family on which to run
$Q$ the same way as above.  Again, let $SFam(\phi_i)$ be the selected
family and let $r_{i,m}$ be the number of rows that $Q$ reads when
running on $S(\phi_i, K_m)$. In addition, let $r$ be the maximum
number of rows that $Q$ can read without exceeding its response time
constraint assuming the sample is stored on the disk. Consider sample
$S(\phi_i, K_q)$ where $K_q$ is the largest value in $SFam(\phi_i)$
that is smaller than $K_m r/r_{i,m}$. Similarly, let $(\phi_i, K'_q)$
be the largest sample of $SFam(\phi_i)$ stored in memory. Finally, we
select the largest sample between $S(\phi_i, K_q)$ and $S(\phi_i,
K'_q)$ and execute $Q_i$ on it.
}

\subsubsection{Queries with Disjunctive Predicates}
\label{sec:disjunctive-predicates}

Consider a query $Q$ with disjunctions in its  {\tt WHERE} clause.
In this case, we rewrite $Q$
as a union of queries $\{Q_1, Q_2,~\ldots, Q_p\}$, where each query
$Q_i$ contains only conjunctive predicates. 
Let $\phi_j$ be the set of
columns in $Q_j$'s predicates. Then, we associate with every query
$Q_i$ an error constraint (\eg standard deviation $s_i$) or time constraint,
such that we can still satisfy $Q$'s error/time constraints when aggregating
the results over $Q_i$ $(1 \leq i \leq p)$ in parallel. 
 Since each of the queries,
$Q_i$ consists of only conjunctive predicates, we select their corresponding
sample families using the selection procedure described in~\xref{sec:conjunctive-predicates}.


\subsection{Selecting the Sample Size}
\label{sec:select-samplesize}

Once a sample family is decided, {\system} needs to select an appropriately
sized sample in that family based on the query's response time or error
constraints. We accomplish this by constructing an {\it Error-Latency Profile}
(ELP) for the query. 
The ELP characterizes the rate at which the error decreases (and the query response
time increases) with increasing sample sizes, and is built simply by running the query on smaller
samples to estimate the selectivity and project latency and error for larger samples. For a distributed query,
its runtime scales with sample size, with the scaling rate depending on the exact
query structure ({\tt JOINS, GROUP BYs} etc.),
physical placement of it's inputs and the underlying data distribution~\cite{rope}.
As shown in Table~\ref{tab:closedform}, the variation of error
(or the variance of the estimator)  primarily
depends on the variance of the underlying data distribution and the actual
number of tuples processed in the sample, which in turn depends on the
selectivity of a query's predicates.


\eat{
For queries with time constraints, we assume linear scaling with
increasing data. This is reasonable given the operations we support,
however the initial analysis must be run on samples large enough to
make sure they are larger than a base amount, with the base amount
determined by processor cache sizes, RAM sizes, and buffer sizes. By
running on multiple instances of data with this smaller size, the
system can determine the average time processing such a query would
take, and linearly scale it up to just below what is allowed by the
time constraints. This scaling provides an ideal sample size, that is
rounded down, allowing the system to satisfy the supplied constraints,
while simultaneously minimizing error. Assuming exponential sample
sizes, with a base size of $k$, the scaling will produce a sample with
no fewer than $\frac{1}{k}$th the number of rows, and hence an error
which is only $\sqrt{k}$ higher than expected.
 }

\vspace{.1in}
\noindent
\textbf{Error Profile:} An error profile is created for all queries with error
constraints. If $Q$ specifies an error (\eg standard deviation) constraint,
the {\system} error profile tries to predict the size of the smallest sample
that satisfies $Q$'s error constraint.
\eat{Assume this sample consists of
a set of randomly selected rows from the original table that match $Q$'s
filter predicates.}
Table~\ref{tab:closedform} shows the formulas of the
variances for the most common aggregate operators. Note that in all
these examples, the variance is proportional to $\sim 1/n$, and thus
the standard deviation (or the statistical error) is proportional to
$\sim 1/\sqrt{n}$, where $n$ is the number of rows from
a sample of size $N$ that match $Q$'s filter predicates. The
ratio $n/N$ is called the _selectivity_ $s_q$ of the query. 

 Let 
$n_{i, m}$ be the number
of rows selected by $Q$ when running on the smallest sample of the selected family,
$S(\phi_i, K_m)$.  Furthermore,  \system{} estimates the query selectivity $s_q$, sample variance
$S_n$ (for {\tt Avg/Sum}) and the input data distribution $f$ (for {\tt Quantiles}) as it runs 
on this sample.
Using these parameter estimates, we calculate
the number of rows $n = n_{i,m}$ required to meet $Q$'s error constraints using the
equations in Table~\ref{tab:closedform}. Then we select the sample
$S(\phi_i, K_q)$ where $K_q$ is the smallest value in $SFam(\phi)$
that is larger than $n*(K_m/n_{i,m})$. This ensures that the expected number of rows
selected by $Q$ when running on sample $S(\phi_i, K_q)$ is $\geq
n$. As a result, the answer of $Q$ on $S(\phi_i, K_q)$ is expected to meet
$Q_i$'s error constraint.

\vspace{.1in} 
\noindent{\textbf{Latency Profile:}} Similarly, a latency profile is
created for all queries with response time constraints. If $Q$ specifies a
response time constraint, we select the sample family on which to run
$Q$ the same way as above.  Again, let $SFam(\phi_i)$ be the selected
family and let $n_{i,m}$ be the number of rows that $Q$ reads when
running on $S(\phi_i, K_m)$. In addition, let $n$ be the maximum
number of rows that $Q$ can read without exceeding its response time
constraint.

$n$ depends on the physical placement of input data
(disk vs. memory), the query structure and complexity, and the degree of
parallelism (or the resources available to the query).
As a simplification, {\system} simply predicts $n$ by assuming latency scales linearly with input size
input data, as is commonly done in parallel distributed execution environments~\cite{mantri-osdi, late-osdi}. 
To avoid non-linearities that may arise when running on very small in-memory samples,
\system{} runs a few smaller samples until performance seems to grow linearly
and then estimates appropriate linear scaling
constants (\ie {\it data processing rate(s), disk/memory I/O rates etc.})
for the model.
These constants are used to estimate a value of 
$n$ that is just below what is allowed by the time constraints. Once $n$ is
estimated, {\system} picks sample $S(\phi_i, K_q)$
where $K_q$ is the largest value in $SFam(\phi_i)$
that is smaller than $n*(K_m /n_{i,m})$ and executes $Q$ on it in parallel. 

\eat{
Similarly, let $(\phi_i, K'_q)$
be the largest sample of $SFam(\phi_i)$ stored in memory. Finally, we
select the largest sample between $S(\phi_i, K_q)$ and $S(\phi_i,
K'_q)$ and execute $Q_i$ on it.
}

\begin{figure}[htbp]
\begin{center}
\includegraphics*[width=225pt]{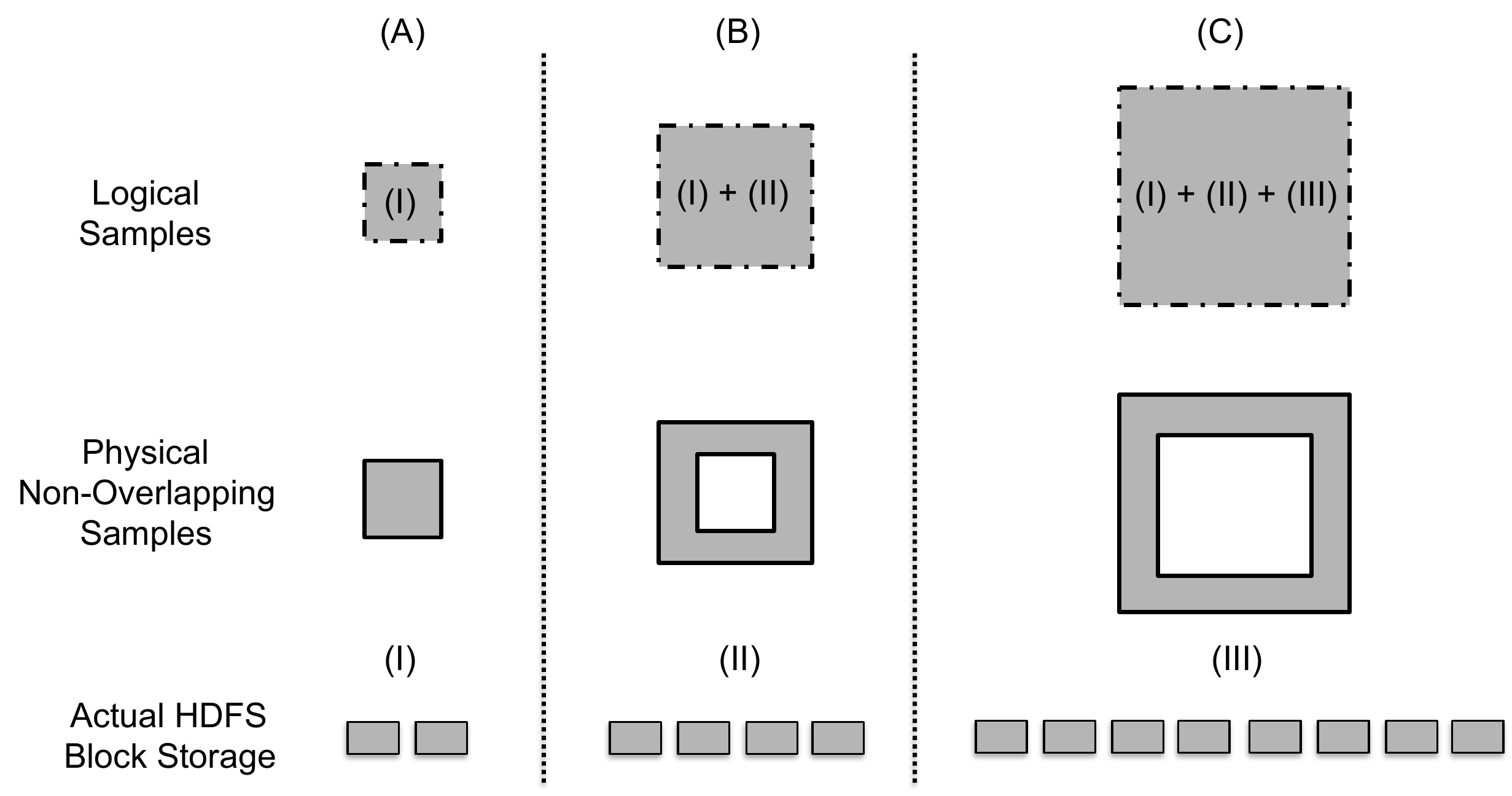}
\caption{Mapping of {\system}'s non-overlapping samples to HDFS blocks}
\label{fig:sampleselection}
\end{center}
\vspace{-.1in}
\end{figure}


\subsection{Query Answers from Stratified Samples}
\label{sec:unbiased-biased}

Consider the \emph{Sessions} table, shown in Table~\ref{tab:toy-table}, and the following query against this 
table.
{\small
\begin{verbatim}
SELECT City, SUM(SessionTime)
FROM Sessions
GROUP BY City
WITHIN 5 SECONDS
\end{verbatim}
}

If we have a uniform sample of this table, estimating the query answer is straightforward. For instance, suppose we take a 
uniform sample with 40\% of the rows of the original \emph{Sessions} table.  In this case, we simply scale the
final sums of the session times by $1/0.4=2.5$ in order to produce an unbiased estimate of the true answer\footnote{Here we use the terms _biased_
and _unbiased_ in a statistical sense, meaning that although the estimate might vary from 
the actual answer, its _expected value_ will be the same as the actual answer.
}.

Using the same approach on a stratified sample may produce a biased estimate of the answer for this query.
For instance, consider a  stratified sample of the \emph{Sessions} table on the \emph{Browser} column, as shown in Table~\ref{tab:toy-sample}.
Here, we have a cap value of $K=1$, meaning we keep all rows whose \emph{Browser} only appears once
in the original \emph{Sessions} table (e.g., \emph{Safari} and \emph{IE}), but when a browser has more than one row (\ie \emph{Firefox}), 
only one of its rows is chosen, uniformly at random.  In this example we have choose the row that 
corresponds to \emph{Yahoo.com}.  Here, we cannot simply scale the final sums of the session times because different values were sampled with different rates. Therefore, to produce unbiased answers, 
\system~ keeps track of the effective sampling rate applied to each row, e.g. 
in Table~\ref{tab:toy-sample}, this rate is $0.33$ for \emph{Firefox} row, while it is $1.0$ for \emph{Safari} and \emph{IE} rows
since they have not been sampled at all. Given these per-row sample rates, obtaining an
unbiased estimates of the final answer is straightforward, e.g., in this case the sum of sessions times  
 is estimated as $\frac{1}{0.33}*20+\frac{1}{1}*82$ for \emph{New York} and  as $\frac{1}{1}*22$ for \emph{Cambridge}. 
 Note that here we will not produce any output for \texttt{Berkeley} (this would not happen if we had access to a stratified sample over \emph{City}, for example). In general, the query processor in \system~ performs a similar correction when operating
on stratified samples.

\begin{table}[!h]
{\small
\hfill{}
\begin{tabular}{| c | c | c | c |}
\hline
\bf URL & \bf City & \bf Browser & \bf SessionTime\\
\hline
cnn.com & New York & Firefox & 15 \\
\hline
yahoo.com & New York & Firefox & 20 \\
\hline
google.com & Berkeley & Firefox & 85\\
\hline
google.com & New York & Safari & 82 \\
\hline
bing.com & Cambridge & IE & 22 \\
\hline
\end{tabular}
\hfill{}
}
\caption{\texttt{Sessions} Table.}
\label{tab:toy-table}
\end{table}

\begin{table}[!h]
\vspace{-.2in}
{\small
\hfill{}
\begin{tabular}{| c | c | c | c | c | }
\hline
\bf URL & \bf City & \bf Browser  & \bf SessionTime & \bf SampleRate\\
\hline
yahoo.com & New York & Firefox & 20 & 0.33 \\
\hline
google.com & New York & Safari & 82 & 1.0\\
\hline
bing.com & Cambridge & IE & 22 & 1.0\\
\hline
\end{tabular}
\hfill{}
}
\caption{A sample of \texttt{Sessions} Table stratified on \texttt{Browser} column.}
\label{tab:toy-sample}

\end{table}

\subsection{Re-using Intermediate Data}
\label{sec:reusing}
Although {\system} requires a query to operate
on smaller samples to construct its ELP, the
intermediate data produced in the process is effectively utilized when
the query runs on larger samples. Fig.~\ref{fig:sampleselection} decouples
the logical and physical view of the non-overlapping
samples maintained by {\system} as described in~\xref{sec:stratified-samples}.
Physically, each progressively bigger logical sample ($A$, $B$ or $C$) consists
of all data blocks of the smaller samples in the same family. {\system} maintains a transparent mapping between
logical samples and data blocks, \ie $A$ maps to (I), $B$ maps to (I, II) and
$C$ maps to (I, II, III). Now, consider a query $Q$ on this data. First, {\system} creates
an ELP for $Q$ by running it on the smallest sample $A$,
\ie it operates on the first two data blocks to estimate various query parameters
described above and caches all intermediate data in this process. Subsequently,
 if sample $C$ is chosen based on the $Q$'s error/latency requirements,
{\system} only operates on the additional data blocks, utilizing the
previously cached intermediate data.

\subsection{Sample Maintenance}\label{sec:sample-maintenance}

{\system}'s reliance on offline sampling
can result in situations where a sample is
not representative of the underlying data. Since
statistical guarantees are given across repeated 
resamplings, such unrepresentative samples can
adversely effect decisions made using \system. 
Such problems are unavoidable when using offline sampling,
and affect all systems relying on such techniques.

As explained in \xref{implementation}, \system~uses a
parallel binomial sampling framework to generate samples when
data is first added. We rely on the same framework for sample
replacement, reapplying the process to existing data, and replacing
samples when the process is complete.

To minimize the overhead of such recomputation, \system~uses a low-priority,
background task to compute new samples from existing data. The task is designed
to run when the cluster is underutilized, and is designed to be suspended at other
times. Furthermore, the task utilizes no more than a small fraction of unutilized
scheduling slots, thus ensuring that any other jobs observe little or no overhead.

\section{Implementation}\label{implementation}

Fig.~\ref{fig:implementation} describes the entire {\system} ecosystem.
{\system} is built on top of the Hive Query Engine~\cite{hive}, supports
both Hadoop MapReduce~\cite{hadoopmr} and Spark~\cite{spark} (via Shark~\cite{shark})
at the execution layer and uses the Hadoop Distributed File System~\cite{hdfs}
at the storage layer.

Our implementation required changes in a few key components. We
added a shim layer of {\it {\systeminitalics} Query Interface} to the HiveQL parser
that enables queries with response time and error bounds. Furthermore, it
detects data input, which causes the {\it Sample Creation and Maintenance}
module to create or update
the set of random and multi-dimensional samples at multiple granularities
as described in \xref{solution:sampling}. We further extend
the HiveQL parser to implement a {\it Sample Selection} module that
re-writes the query and iteratively assigns it an appropriately sized
biased or random sample as described in \xref{solution:selection}.
We also added an {\it Uncertainty Propagation} module to modify
the pre-existing aggregation functions summarized in Table~\ref{tab:closedform}
to return errors bars and confidence intervals in addition to the result. Finally,
we extended the SQLite based Hive Metastore to create {\it {\systeminitalics}
Metastore} that maintains a transparent mapping between the non-overlapping logical
samples and physical HDFS data blocks as shown in Fig.~\ref{fig:sampleselection}.


\begin{figure}[tbp]
\begin{center}
\includegraphics*[width=150pt]{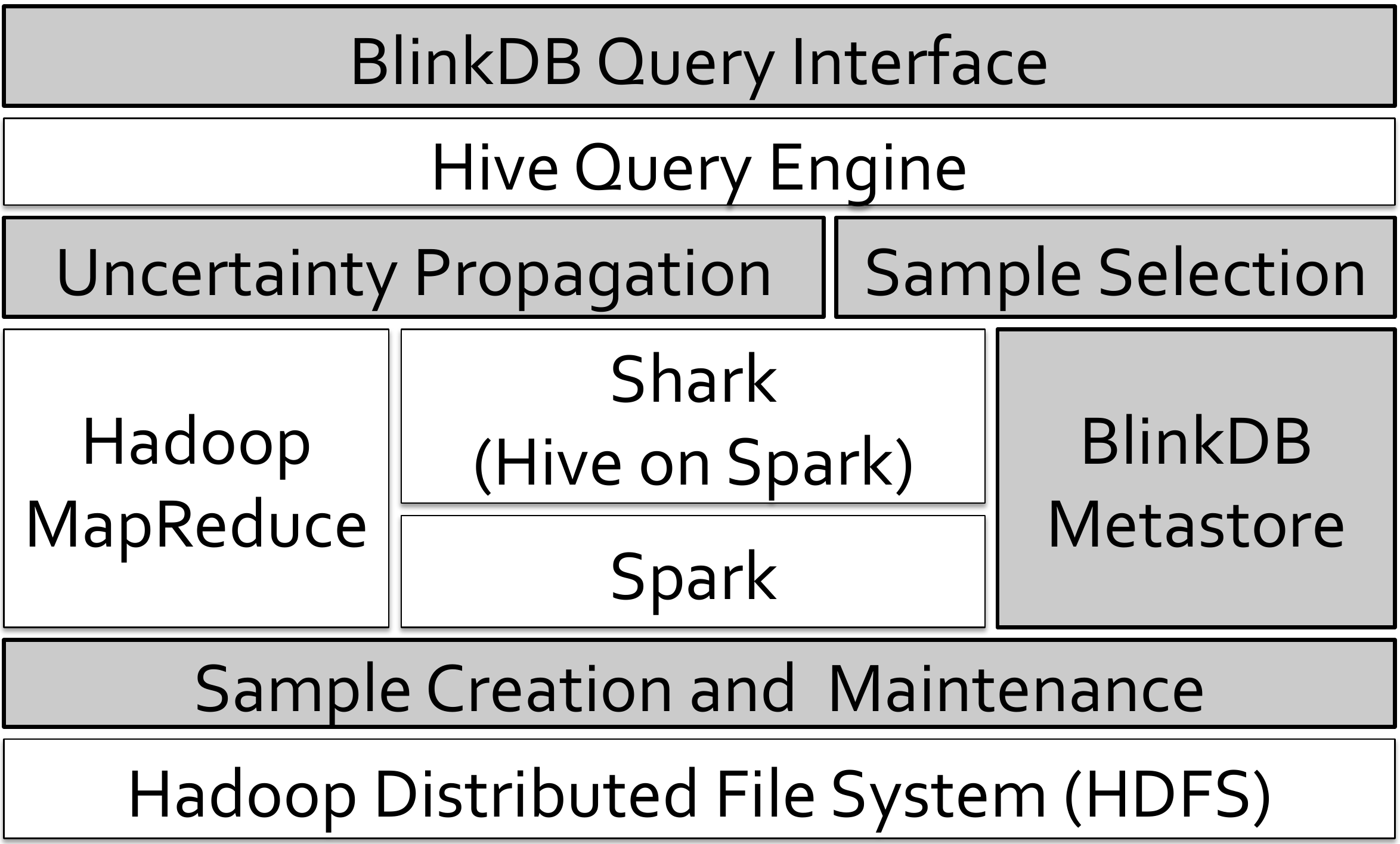}
\caption{\system's Implementation Stack}
\label{fig:implementation}
\vspace{-0.3in}
\end{center}
\end{figure}

We also extend Hive to add support for sampling from tables.
This allows us to leverage Hive's parallel execution engine for sample creation in a 
distributed environment. Furthermore, our sample
creation module optimizes block size and data placement for samples in
HDFS.

\eat{We extended Hive with support for sampling;  this allows us to build samples
by leveraging the Hive parallel query execution engine.}
\if{0}
In order to incur little overhead in creating and replacing samples as new
data in being added in the system, \system{} has parallel sample
creation framework that creates in-place binomial samples for data stored in
HDFS. This is achieved by leveraging the Hive query execution engine to
create and maintain samples in {\system}.

Specifically, we augmented the
HiveQL to add {\tt SAMPLE ON [RANDOM | Column Name(s)]} operators. The
{\tt SAMPLE ON RANDOM} operator shares the same semantics of a {\tt WHERE}
clause and outputs each input row with a probability $p$. The parameterized {\tt SAMPLE
ON Column Name(s)} operator takes one or more columns parameters and aggregates
the input data on the unique values (_key_) in those column combinations. It then applies the
{\tt SAMPLE ON RANDOM} operator on each unique _key_ to create stratified samples.
\fi
In \system{}, uniform samples are generally created in a few hundred
seconds. This is because the time taken to create them only depends on the disk/memory bandwidth
and the degree of parallelism. On the other hand, creating stratified
samples on a set of columns takes anywhere between a $5-30$ minutes
depending on the number of unique values to stratify on, which decides the number
of reducers and the amount of data shuffled.
\eat{Please note that using binomial
sampling in place of reservoir sampling does introduce some statistical bias in the
system which is appropriately corrected.}

\if{0}
\subsection{Storage Optimizations}
\label{sec:data-partition}

Query response times in \system~are usually dominated by the time
to access stored samples on the disk.
As such, there are two storage-related implementation questions that
have a significant impact on the sample access time: (1) What is the
size of the file system block?, and (2) How is a sample stored in the
underlying file system? We answer these questions next in the context
of a HDFS-like file system.

Smaller block sizes lead to higher parallelism as HDFS can do a better
job spreading a file across the nodes in a cluster. On the other
hand, large blocks reduce file system overhead~\cite{namenode-pressure},
and improves the disk throughput. In our implementation, we balance
this trade-off by picking $3$ MB blocks, such that each block can be
read in a fraction of a second. While these blocks are significantly smaller than the
ones used in most HDFS deployments, our experience so far has not
revealed a significant impact on the read throughput and the system
scalability for our deployment.

\eat{
To answer the first question, consider a cluster consisting of $M$
disk drives, and let $R$ be the disk throughput. Then, we can read a
file of size $X$ in as little as $X/(M R)$ time. This assume that the
blocks of the file are perfectly load balanced across all disks.
To achieve a good level of load balancing across a large variety of
file sizes we want the blocks to be as small as
possible. Unfortunately, small blocks have a negative impact on disk
throughput, and incur a high file system overhead. For these reasons
and because they target batch workloads, the typical deployments of
HDFS use block sizes ranging from $64$ MB, all the way up to $1$
GB. Unfortunately, such large blocks would significantly impact query
response times. Assuming the sequential throughput of a disk is $R =
50$ MB/s, the time it takes to read a sample larger than $1$ GB is
bounded below by $1.3$s for $64$ GB blocks, and as much as $20$s for
$1$ GB blocks. As a result, in our deployment, we use much smaller
blocks of a few MBs. While these blocks are significantly smaller than
the ones used in most HDFS deployments, our experience so far has not
revealed a significant impact on the read throughput and the system
scalability.  }

We consider two answers to the second question: maintain the entire
sample in a single file or partition the sample in many small
files. First, assume we store the sample in a single file. Recall that
a stratified sample, $S(\phi, K)$, is clustered by the values in
$\phi$. In order to take advantage of this clustering, we need to know
at which location in the file does a key starts. This requires us to
maintain metadata associated to each file. The alternative is to split
the sample into smaller files along the key values in $\phi$. In this
case, the query plan will identify the file(s) in which the keys are
stored and then execute the query only on these files. We picked
the second approach for it significantly reduces the
{\it {\systeminitalics} Metastore}'s complexity.

\eat{
Since
maintaining non-overlapping samples requires us to split a sample in
multiple files anyway, in this paper we use the second approach and
map each sample to many small files. In particular, we pick the size
of a file to be roughly $B \times M$, where $B$ is the block size. In
the best case scenario, this design allows us to read a file as fast
as reading a block. This approach simplifies the design as it does not
require to maintain additional metadata per file, as in the case of
storing each sample in a large file.
}

\begin{figure}[htbp]
\begin{center}
\includegraphics*[width=225pt]{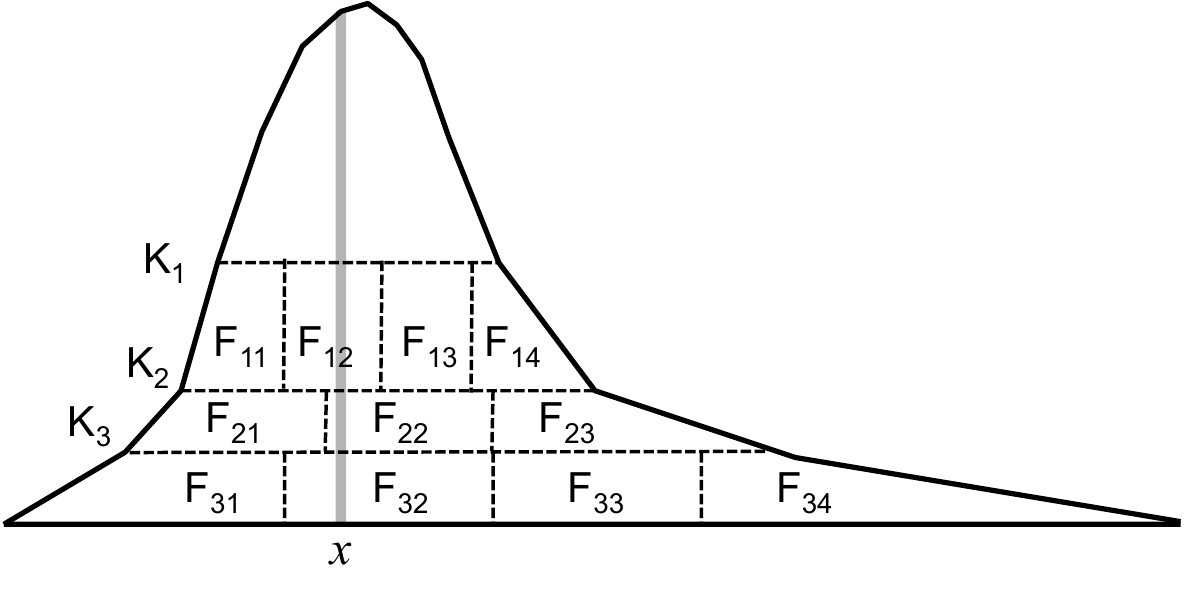}
\caption{Example of mapping the stratified sample in
  Figure~\ref{fig:stratified-sample} to many small files.}
\label{fig:sample2file-mapping}
\end{center}
\end{figure}

Figure~\ref{fig:sample2file-mapping} shows an example where the sample
$S(\phi, K_1)$ in Figure~\ref{fig:stratified-sample} is partitioned
into small files. In particular, the lower row, corresponding to
$S(\phi, K_3)$, is partitioned into files $\{F_{31}, F_{32}, F_{33},
F_{34}\}$. Similarly, the second row, which together with the first
row correspond to $S(\phi, K_2)$, is partitioned into files $\{F_{21},
F_{22}, F_{23}\}$. Finally, the last row, which together with the
previous two rows correspond to $S(\phi, K_1)$, is partitioned into
$\{F_{11}, F_{12}, F_{13}, F_{14}\}$. Now consider query $Q$ whose
{\tt WHERE} clause is $(\phi = x)$. If we use sample $S(\phi, K_3)$ to
answer $Q$, then $Q$ needs to read only file $F_{3, 2}$. If we use
sample $S(\phi, K_2)$, then $Q$ will read both $F_{3,2}$ and $F_{2,
  2}$. Finally, if we use the largest sample, $S(\phi, K_1)$, then $Q$
will read all files that contain value $x$,~\ie~$F_{3,2}$, $F_{2,2}$,
and $F_{1, 2}$, respectively.
\fi

\section{Evaluation}\label{evaluation}
In this section, we evaluate {\system}'s performance on a $100$ node EC2 cluster
using two workloads: a workload from Conviva Inc.~\cite{Conviva}
and the well-known TPC-H  benchmark~\cite{tpch}.
First, we compare {\system} to query execution on full-sized
datasets to demonstrate how even a small trade-off in the accuracy of final answers can
result in orders-of-magnitude improvements in query response times. Second, we evaluate
the accuracy and convergence properties of our optimal multi-dimensional, multi-granular
stratified-sampling approach against both random
sampling and single-column stratified-sampling approaches. Third, we evaluate the effectiveness
of our  cost models and error projections at meeting the user's accuracy/response time requirements.
Finally, we demonstrate {\system}'s ability to scale gracefully with increasing cluster size.

\eat{We evaluate a few different aspects of \system~: (i) The performance of \system~versus other approaches, e.g., a) Native Hive running on both MapReduce and Spark, b) Sampling approaches that do not pre-compute samples (\eg online aggregation) and c) Sampling approaches based on single-column stratified samples rather than an optimally chosen set of multi-dimensional ones along with (ii) The ability of \system~to meet the user's accuracy/response time requirements. We report  experiments on two different sets of workloads: (1) A real-world dataset and query logs from Conviva Inc.~\cite{Conviva} and (2)  the TPC-H benchmark. The Conviva and the TPC-H datasets were $17$ TB and $1$ TB (i.e., a scale factor of $1000$) in size respectively and were stored simultaneously across $100$ Amazon EC2 extra large nodes\footnote{Amazon EC2 extra large nodes have  $8$ CPU cores ($2.66$ GHz), $68.4$ GB of RAM, with an instance-attached disk of $800$ GB.}. The cluster is configured to utilize $75$ TB of distributed disk storage and $6$ TB of distributed RAM cache.}

\subsection{Evaluation Setting}

The Conviva and the TPC-H datasets were $17$ TB and $1$ TB (\ie a scale factor of $1000$)
in size, respectively, and were both stored across $100$ Amazon EC2 extra large
instances (each with  $8$ CPU cores ($2.66$ GHz), $68.4$ GB of RAM,
and $800$ GB of disk). The cluster was configured to utilize $75$ TB of
distributed disk storage and $6$ TB of distributed RAM cache.\notesameer{should we highlight
here that we are evaluating other frameworks against an implementation of \system{} on spark only
and cite shark paper if readers are interested in hive vs. shark comparison?}

\begin{figure*}[ht]
\vspace{-.2in}
\centering
\subfigure[Sample Families (Conviva)]{
\includegraphics*[width=160pt]{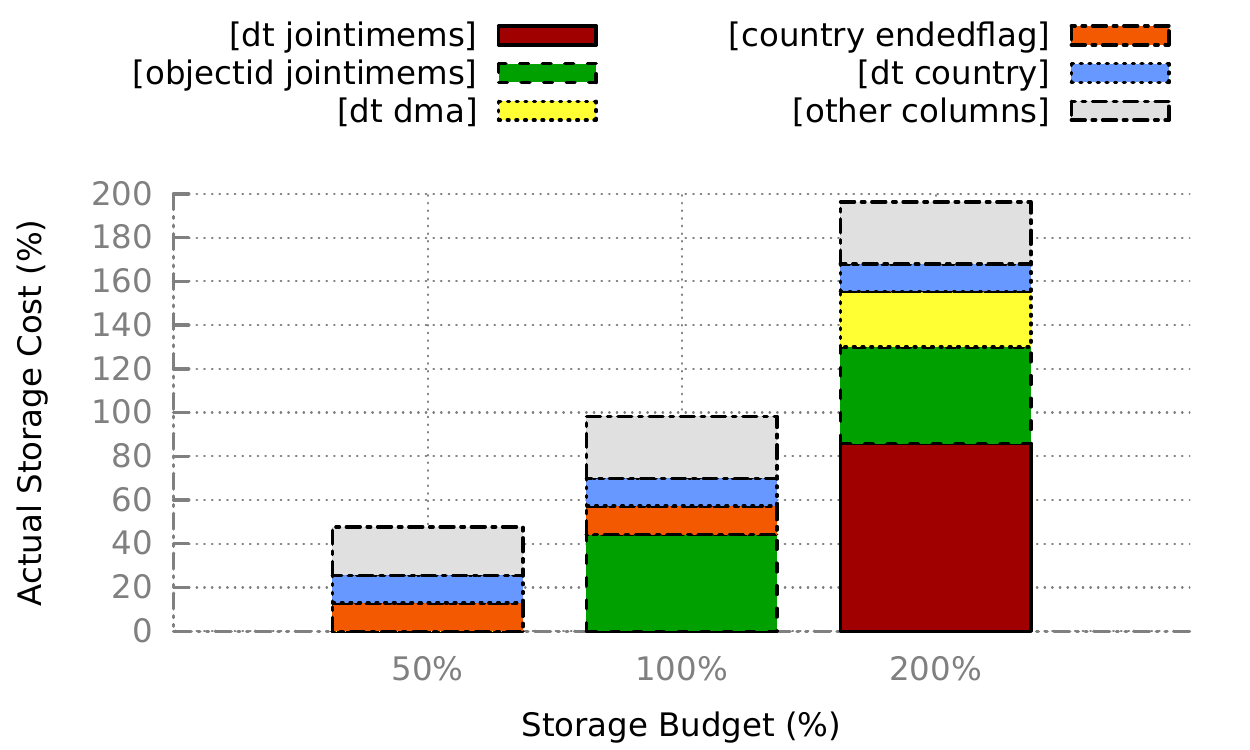}
	\label{fig:storageconviva}
}
\subfigure[Sample Families (TPC-H)]{
\includegraphics*[width=160pt]{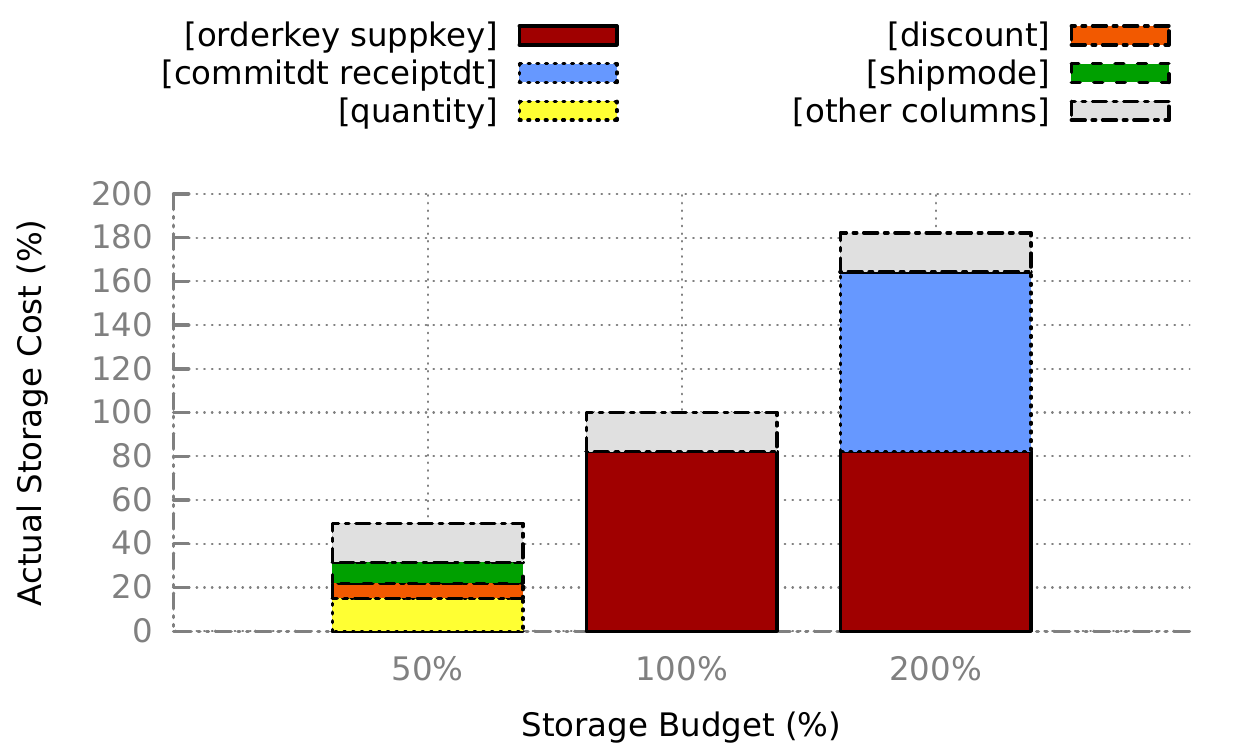}
	\label{fig:storagetpch}
}
\subfigure[{\system} Vs. No Sampling]{
\includegraphics*[width=160pt]{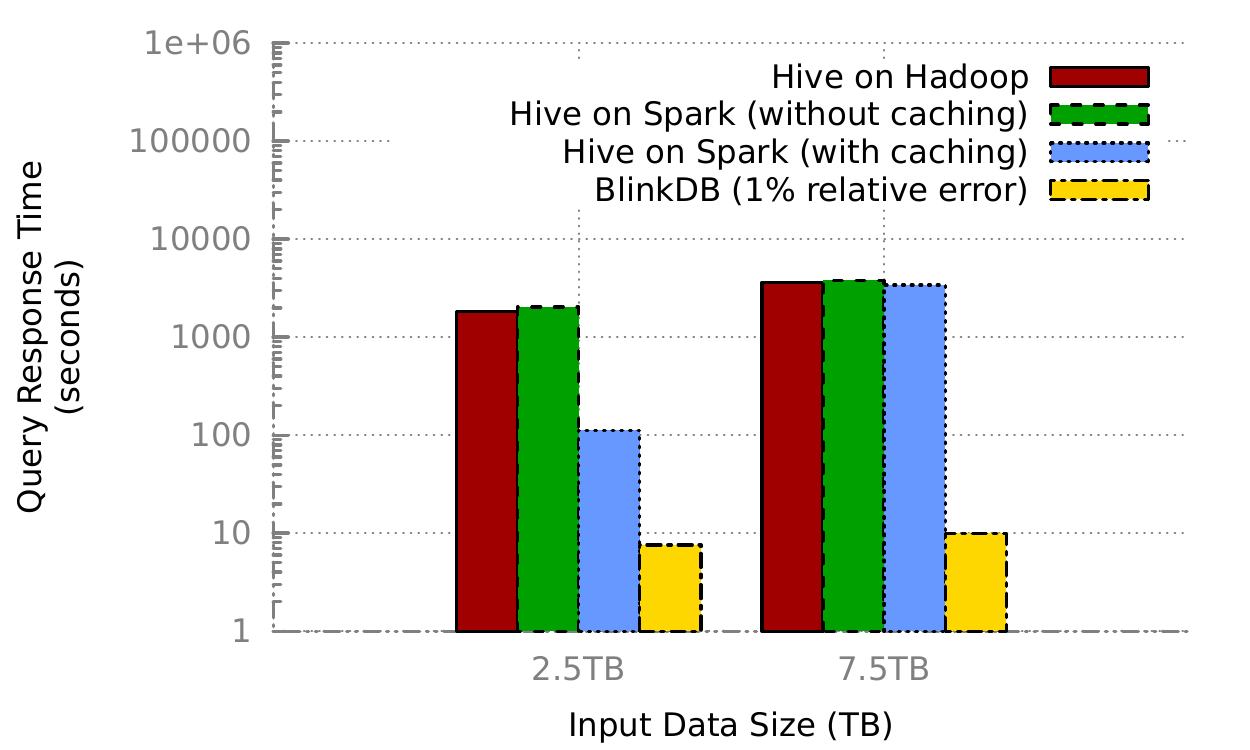}
	\label{fig:qs-vs-hive-vs-shark}
}
\vspace{-.1in}
\caption{\ref{fig:storageconviva} and~\ref{fig:storagetpch} show the
relative sizes of the set of stratified sample(s) created
for $50$\%, $100$\% and $200$\% storage budget on Conviva
and TPC-H workloads respectively.~\ref{fig:qs-vs-hive-vs-shark}
compares the response times (in log scale) incurred by Hive
(on Hadoop), Shark (Hive on Spark) -- both with and without
input data caching, and {\system}, on simple aggregation.}
\label{fig:optimization}
\vspace{-.1in}
\end{figure*}


\vspace{.1in}
\noindent
\textbf{Conviva Workload.}
The Conviva data represents information about video streams viewed by Internet users.
We use query traces from their SQL-based  ad-hoc querying system which is used for
problem diagnosis and data analytics on a log of media accesses by Conviva users.
These access logs are $1.7$ TB in size and constitute a small fraction of data collected
across $30$ days. Based on their underlying data distribution, we generated a $17$ TB
dataset for our experiments and partitioned it across $100$ nodes. The data consists of
a single large _fact_ table with $104$ columns, such as, customer ID, city, media URL,
genre, date, time, user OS, browser type, request response time, etc. The $17$ TB
dataset has about $5.5$ billion rows.

The raw query log consists of $19,296$ queries, from which we selected different subsets
for each of our experiments. We ran our optimization function on a sample of about $200$
queries representing $42$ query templates.
 We repeated
the experiments with different storage budgets for the stratified samples-- $50\%,
100\%$, and $200\%$.  A storage budget of $x\%$ indicates that the cumulative size
of all the samples will not exceed $\frac{x}{100}$ times the original data. So, for example,
a budget of $100\%$ indicates that the total size of all the samples should be less than
or equal to the original data.
Fig.~\ref{fig:storageconviva} shows the set of sample families that were selected by our
optimization problem for the storage budgets of $50\%, 100\%$ and $200\%$
respectively, along with their cumulative storage costs. Note that each stratified
sample family has a different size due to variable number of distinct keys in the
columns on which the sample is biased. Within each sample family, each
successive resolution is twice as large than the previous one and the value of
$K$ in the stratified sampling is set to $100,000$.

\vspace{.1in}
\noindent
\textbf{TPC-H Workload.}
We also ran a smaller number of experiments on TPC-H to
demonstrate the generality of our results, with respect to a standard benchmark. All the
TPC-H experiments ran on the same $100$ node cluster, on $1$ TB of data (\ie a scale
factor of $1000$). The $22$ benchmark queries in TPC-H
were mapped to $6$ unique query templates.
Fig.~\ref{fig:storagetpch} shows the set of sample families  selected by our
optimization problem for the storage budgets of $50\%, 100\%$ and $200\%$,
along with their cumulative storage costs.

Unless otherwise specified, all the experiments in this paper are done with a
$50\%$ additional storage budget (\ie samples could use an additional storage of
up to $50\%$ of the original data size).

\subsection{\systemheader{} vs. No Sampling}

We first compare the performance of \system{} versus frameworks that execute queries on
complete data.  In this experiment, we ran on two subsets of the Conviva data, with $7.5$ TB and
$2.5$ TB respectively, spread across $100$ machines. We chose these two subsets 
to demonstrate some key aspects of the interaction between data-parallel
frameworks and modern clusters with high-memory servers. While the smaller $2.5$ TB
dataset can be be completely cached in memory, datasets larger than $6$ TB in size have
to be (at least partially) spilled to disk. To demonstrate the significance of sampling even for
the simplest analytical queries, we ran a simple query that computed {\tt average} of user session
times with a filtering predicate on the date column ($dt$) and a {\tt GROUP BY} on the $city$ column.
We compared the response time of the full (accurate) execution of this query on  Hive~\cite{hive}
on Hadoop MapReduce~\cite{hadoopmr}, Hive on Spark (called Shark~\cite{shark}) -- both with
and without caching, against its (approximate) execution on \system~with a $1\%$ error bound for
each {\tt GROUP BY} key at $95\%$ confidence. We ran this query on both data sizes (\ie corresponding to
$5$ and $15$ days worth of logs, respectively) on the aforementioned $100$-node cluster. We
repeated each query $10$ times, and report the average response time in Figure~\ref{fig:qs-vs-hive-vs-shark}.
Note that the Y axis is log scale. In all cases, \system{} significantly outperforms its counterparts
(by a factor of $10-100\times$), because it is able to read far less data to compute a fairly accurate
answer.  For both data sizes,\eat{response times are just a few seconds in} \system{} returned the
answers in a few seconds as compared to thousands of seconds for others. In the $2.5$ TB run,
Shark's caching capabilities considerably help, bringing the query runtime down to about $112$
seconds. However,  with $7.5$ TB data size, a considerable portion of data is spilled to disk and the
overall query response time is considerably longer.

\begin{figure*}[ht]
\vspace{-.2in}
\centering
\subfigure[Error Comparison (Conviva)]{
\includegraphics*[width=160pt]{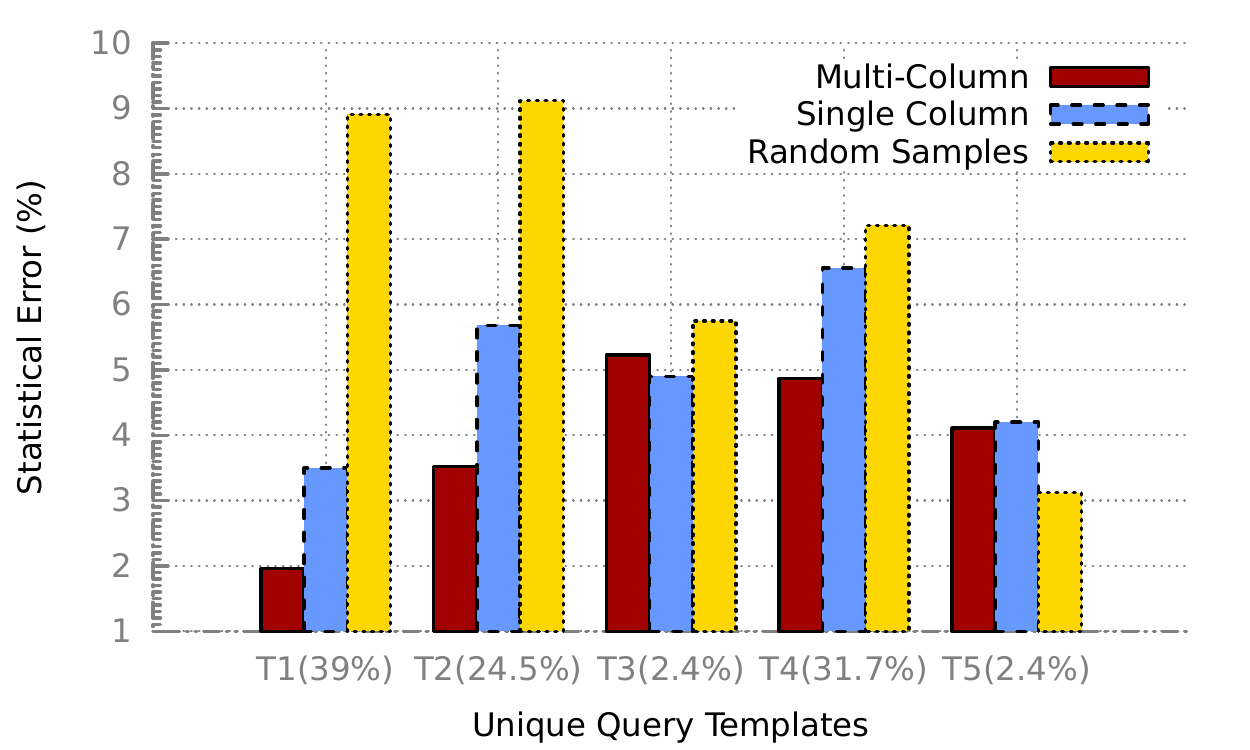}
	\label{fig:opt-conviva}
}
\subfigure[Error Comparison (TPC-H)]{
\includegraphics*[width=160pt]{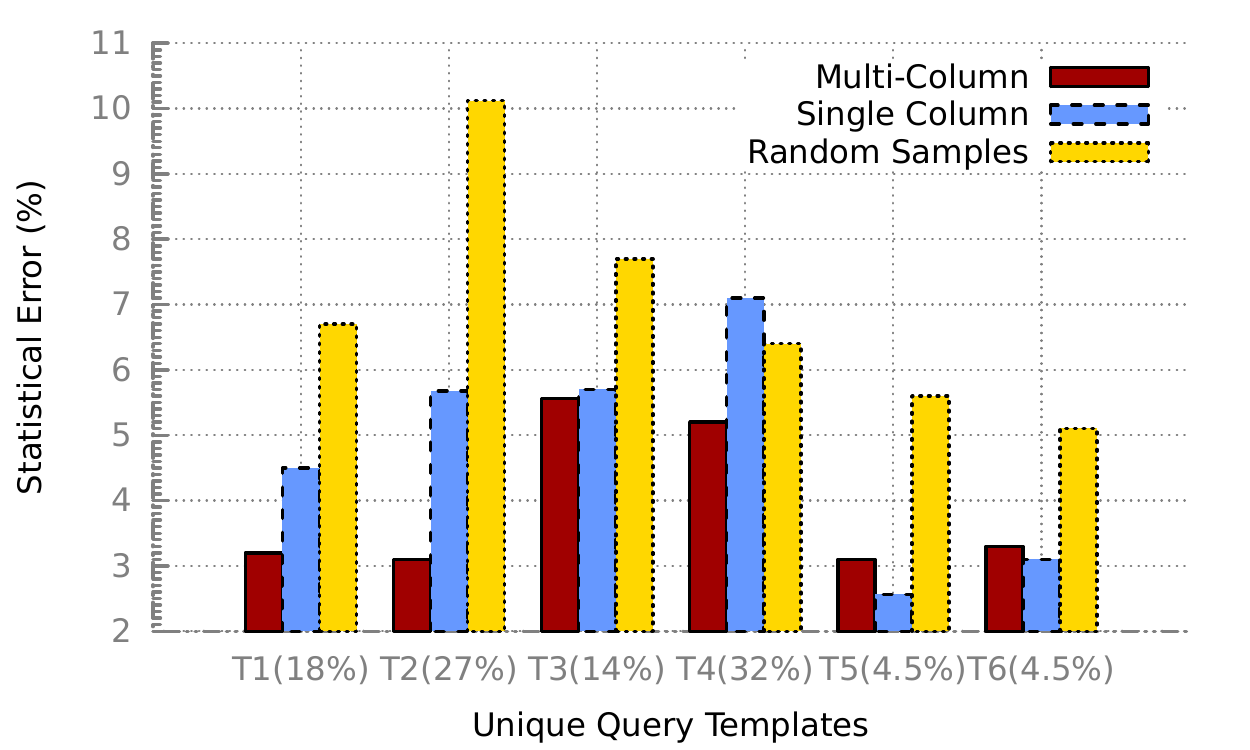}
	\label{fig:opt-tpch}
}
\subfigure[Error Convergence (Conviva) ]{
\includegraphics*[width=160pt]{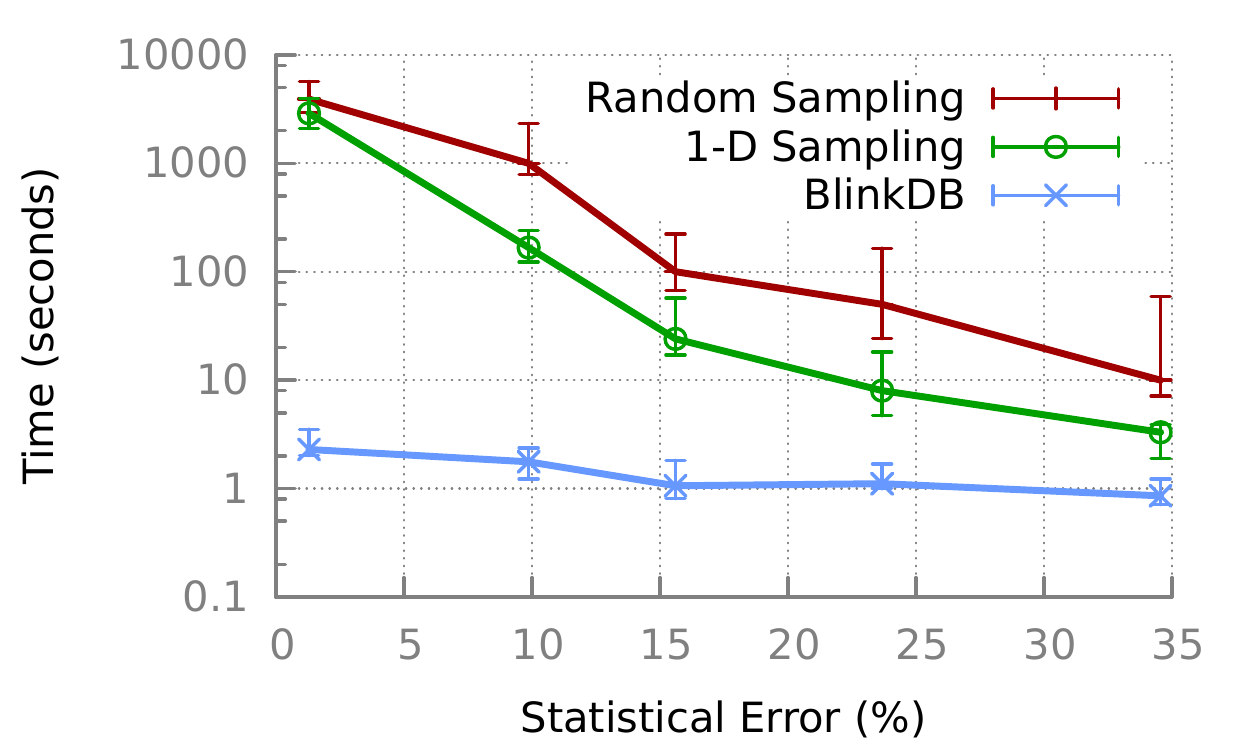}
\label{fig:ola-comparison}
}
\vspace{-.1in}
\caption[]{\ref{fig:opt-conviva} and~\ref{fig:opt-tpch} compare the average statistical
error per template when running a query with fixed time budget for various sets of
samples.~\ref{fig:ola-comparison} compares the rates of error convergence 
with respect to time for various sets of samples.} 
\label{fig:optproblem}
\vspace{-.1in}
\end{figure*}

\subsection{Multi-Dimensional Stratified Sampling}
\label{sec:multi-dimensional-stratified-samples}
Next, we ran a set of experiments to evaluate the error~(\xref{sec:multi-dimension-sampling-exp})
and convergence~(\xref{sec:convergence-exp}) properties of our optimal multi-dimensional,
multi-granular stratified-sampling approach against both simple random
sampling, and one-dimensional stratified sampling (\ie stratified samples over a single column).
For these experiments we constructed three sets of samples on both Conviva and TPC-H
data with a $50\%$ storage constraint:

\begin{asparaenum}
\item \textbf{Multi-Dimensional Stratified Samples}. The sets of columns to stratify on were
chosen using \system{}'s optimization framework (\xref{sec:optimal-view-creation}), restricted
so that  samples could be stratified on no more than $3$ columns (considering four or more 
column combinations caused our optimizer to take more than a minute to complete).
\item \textbf{Single-Dimensional Stratified Samples}. The column to stratify on was chosen using 
the same optimization framework, restricted so a sample is stratified on exactly one column.
\item \textbf{Uniform Samples}. A  sample containing $50\%$ of the entire data, chosen uniformly
at random.
\end{asparaenum}

\subsubsection{Error Properties}
\label{sec:multi-dimension-sampling-exp}

In order to illustrate the advantages of our multi-dimensional stratified sampling strategy,
we compared the average statistical error at $95\%$ confidence while running a query
for $10$ seconds over the three sets of samples, all of which were constrained to be of the 
same size.

For our evaluation using Conviva's data we used a set of
$40$ queries (with $5$ unique query templates) and $17$
TB of uncompressed data on $100$ nodes. We ran a similar
set of experiments on the standard TPC-H queries. The queries
we chose were on the $lineitem$ table, and were modified to
conform with HiveQL syntax. 

In Figures~\ref{fig:opt-conviva}, and \ref{fig:opt-tpch}, we report results per-query template, with numbers
in parentheses indicating the percentage of queries with a given template.
For common query templates, multi-dimensional samples produce smaller statistical errors than either
one-dimensional or random samples. The optimization framework attempts to minimize expected error, rather than per-query 
errors, and therefore for some specific query templates single-dimensional stratified samples behave better than multi-dimensional samples.
Overall, however, our optimization framework significantly improves performance versus single column samples.

\subsubsection{Convergence Properties}
\label{sec:convergence-exp}
We also ran experiments to demonstrate the
convergence properties of multi-dimensional stratified samples used by \system{}.
We use the same set of three samples as \xref{sec:multi-dimensional-stratified-samples}, taken
over $17$ TB of Conviva data. Over this data, we ran multiple queries to calculate average session time
for a particular ISP's customers in $5$ US Cities and determined the latency for achieving a particular
error bound with $95\%$ confidence.  Results from this experiment (Figure~\ref{fig:ola-comparison}) show
that error bars from running queries over multi-dimensional samples converge orders-of-magnitude faster than
 random sampling, and are significantly faster to converge than single-dimensional stratified samples.
\eat{
\begin{figure}[ht]
\begin{center}
\includegraphics*[width=225pt]{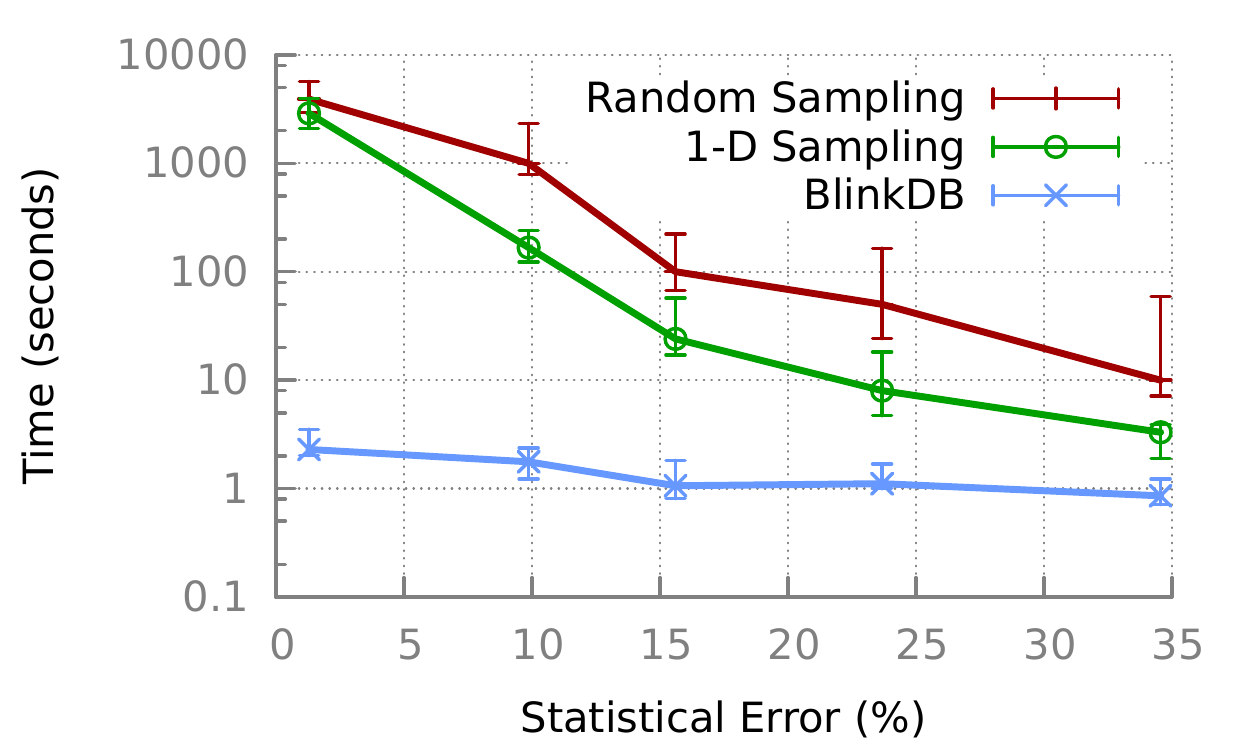}
\caption{Comparison of Convergence Properties.}
\label{fig:ola-comparison}
\end{center}
\end{figure}
}

\subsection{Time/Accuracy Guarantees}
\label{sec:time-accuracy}
In this set of experiments, we evaluate \system{}'s effectiveness at meeting
different time/error bounds requested by the user.  To test time-bounded queries,  we picked a sample of $20$ Conviva queries, and ran each of them $10$ times, with a time bound from $1$ to $10$ seconds. Figure~\ref{fig:time-sla-bounds}  shows the results run on the same $17$ TB data set, where each bar represents the minimum, maximum and average response times of the $20$ queries, averaged over $10$ runs. 
t
From these results we can see that \system{}~is able to accurately select a sample to satisfy a target  response time.

Figure~\ref{fig:error-sla-bounds}  shows results from the same set of queries, also on the $17$ TB data set, evaluating our ability to meet specified error constraints. In this case, we varied the requested error bound from $2\%$ to $32\%$ . The bars  again represent the minimum, maximum and average errors across different runs of the queries. 
Note that the measured error is almost always at or less than the requested error. However, as we increase the error bound, the measured error becomes closer to the bound.  This is because at higher error rates the sample size is quite small and error bounds are  wider. 

\begin{figure*}[ht]
\vspace{-.2in}
\centering
\subfigure[Response Time Bounds]{
\includegraphics*[width=160pt]{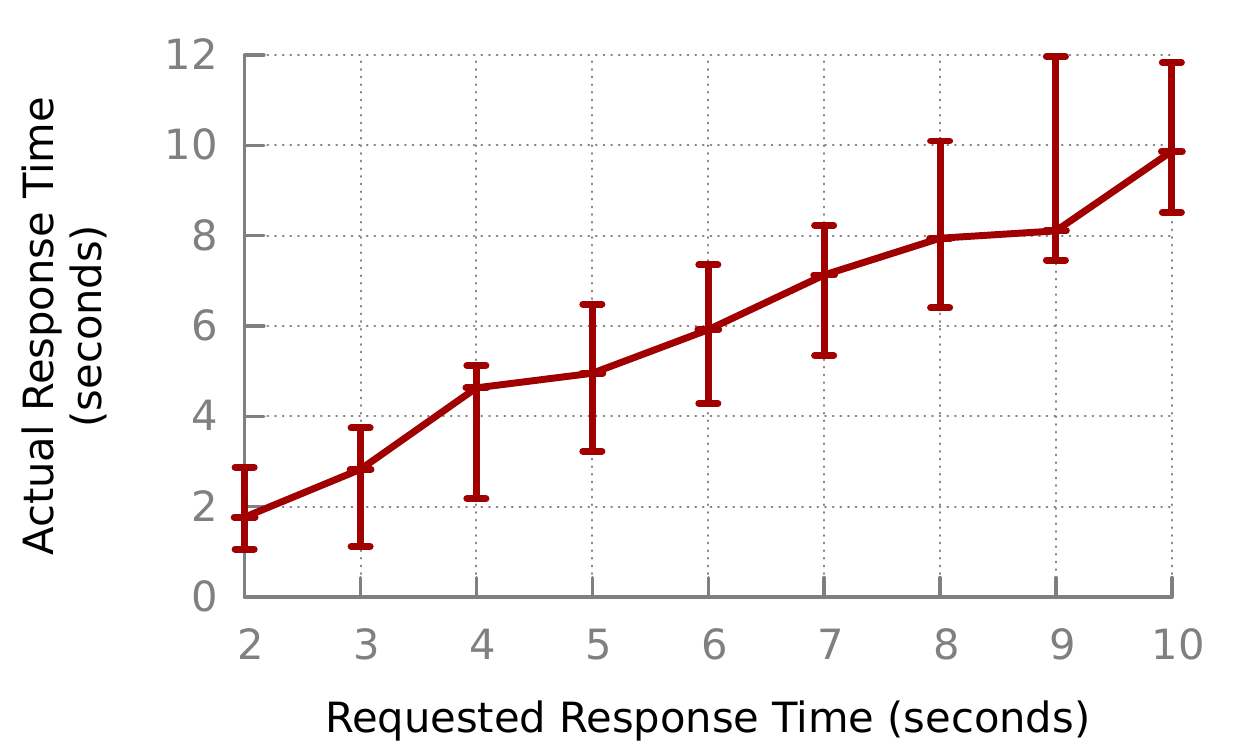}
	\label{fig:time-sla-bounds}
}
\subfigure[Relative Error Bounds]{
\includegraphics*[width=160pt]{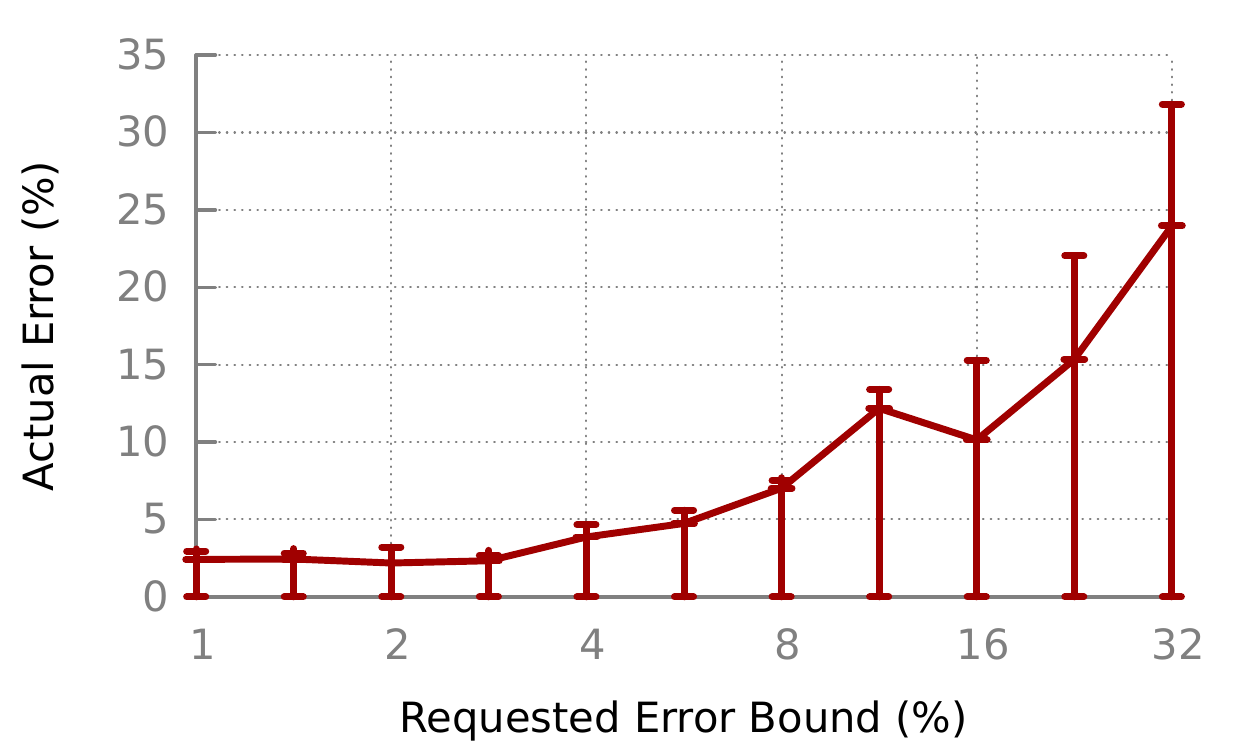}
	\label{fig:error-sla-bounds}
}
\subfigure[Scaleup]{
\includegraphics*[width=160pt]{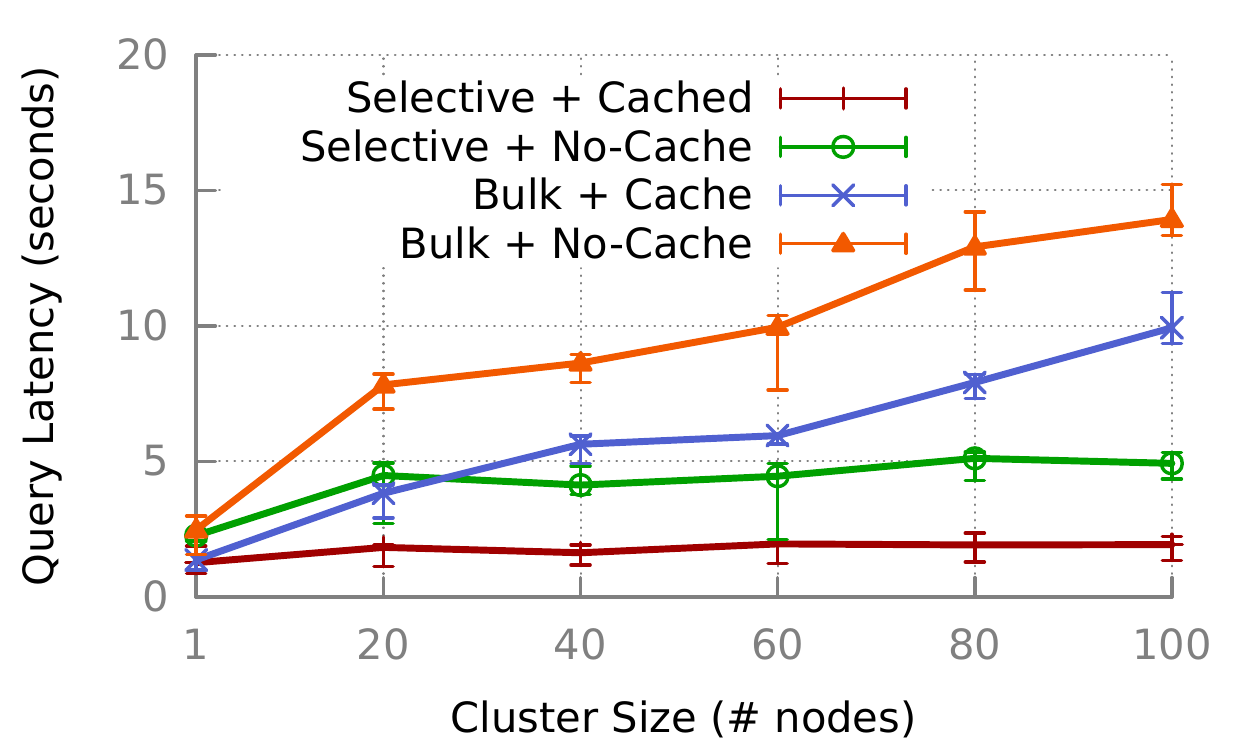}
\label{fig:scaleup}
}
\vspace{-.1in}
\caption[]{\ref{fig:time-sla-bounds} and~\ref{fig:error-sla-bounds} plot the actual vs.
requested response time and error bounds in {\system}.~\ref{fig:scaleup} plots the query
latency across $2$ different query workloads (with cached and non-cached samples)
as a function of cluster size}
\label{fig:bounds}
\vspace{-.1in}
\end{figure*}

\eat{
\begin{figure}[htbp]
\begin{center}
\includegraphics*[width=225pt]{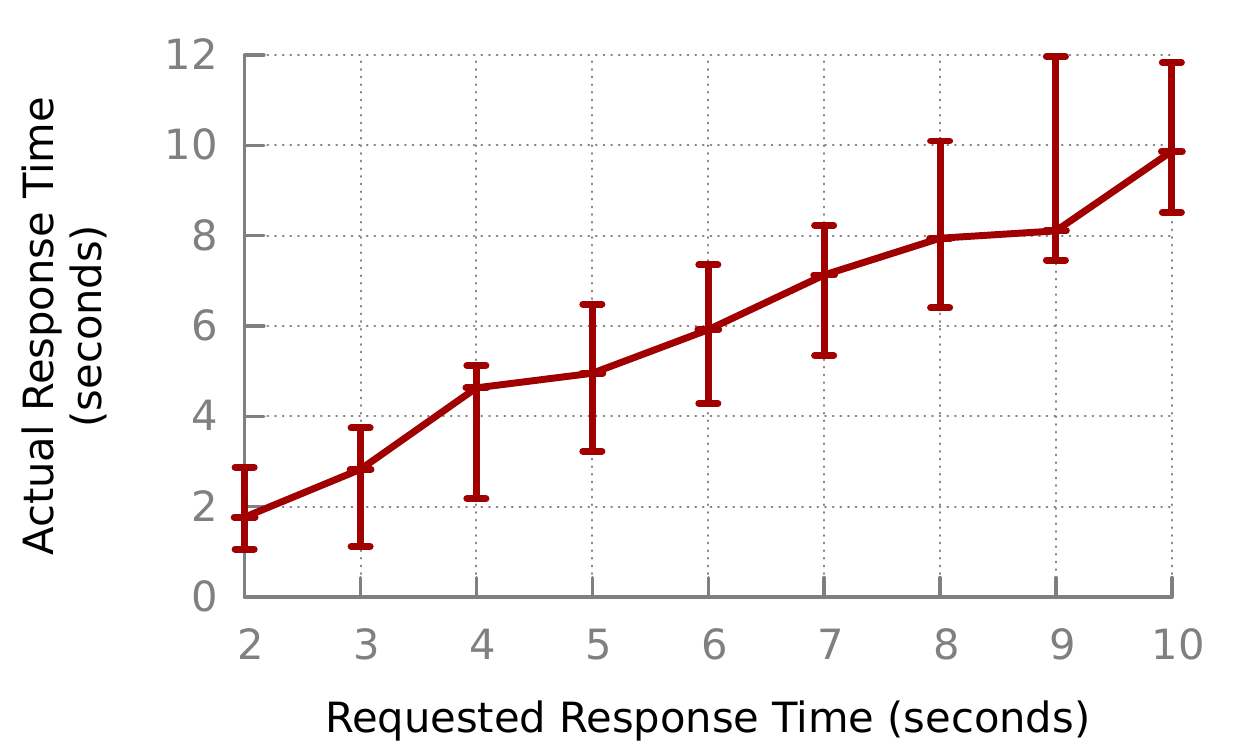}
\caption{Actual vs. requested query response time in {\system}}
\label{fig:time-sla-bounds}
\end{center}
\end{figure}

\begin{figure}[t]
\begin{center}
\includegraphics*[width=225pt]{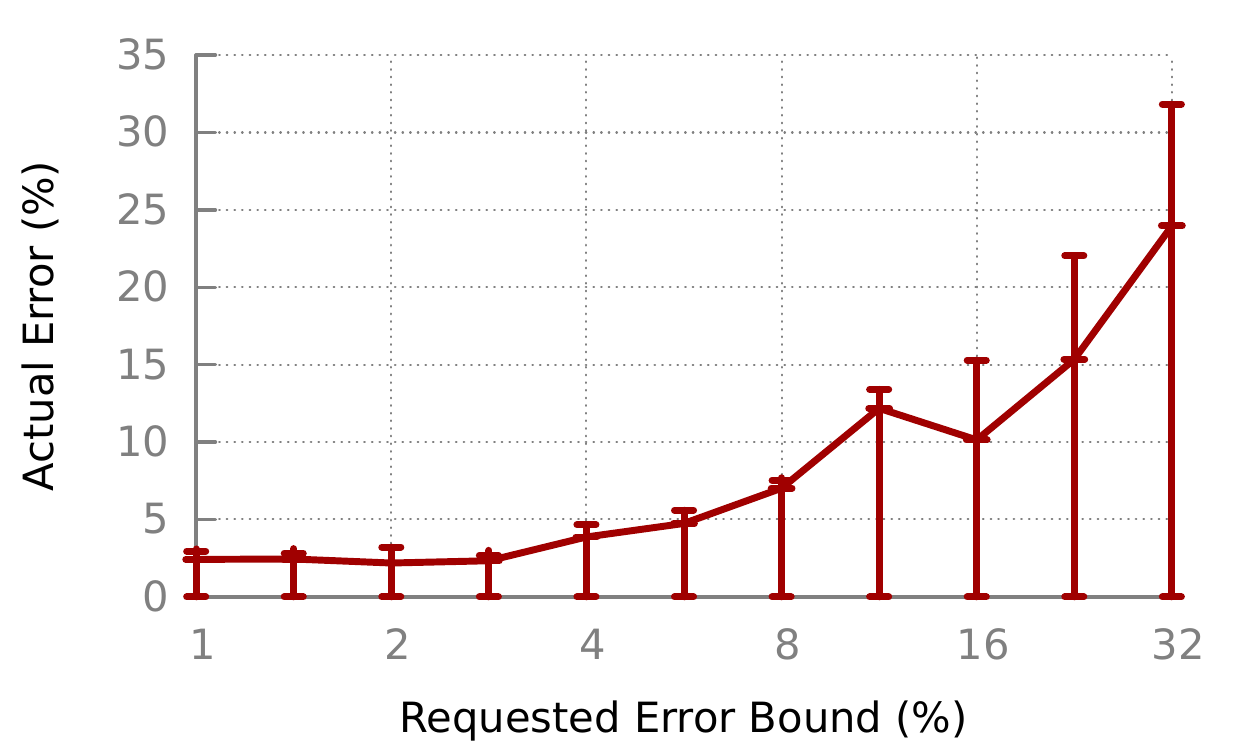}
\caption{Actual vs. requested query error bounds in {\system}}
\label{fig:error-sla-bounds}
\end{center}
\end{figure}
}

\eat{
Figure~\ref{fig:error-vs-time-micro-bench} shows how, for a simple average  operator running 
on the $17$ TB data set, the response time and error bounds varies as we vary the sample size. Other operators
also exhibit a very similar behavior. Note that while error bars are sufficiently wider for smaller samples, the
answers get fairly accurate as the sample size increases.

\eat{
\begin{figure}[t]
\begin{center}
\includegraphics*[width=225pt]{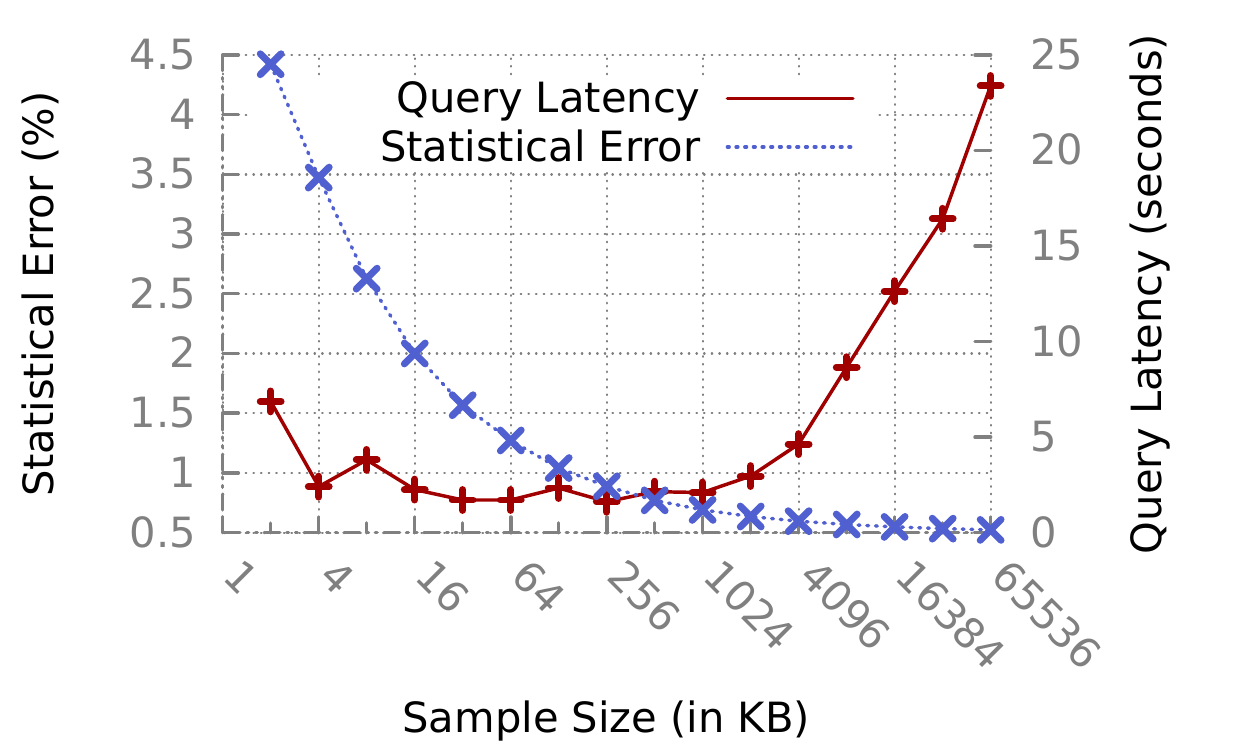}
\caption{This figure depicts the variation of statistical error and response time with respect to sample sizes.}
\label{fig:error-vs-time-micro-bench}
\end{center}
\end{figure}
}
}

\subsection{Scaling Up}
Finally, in order to evaluate the scalability properties of \system{} as a function of cluster size, we created $2$ different sets of query workload suites consisting of $40$ unique Conviva queries each. The first set (marked as $selective$) consists of highly selective queries -- i.e., those queries that only operate on a small fraction of input data. These queries occur frequently in production workloads and consist of one or more highly selective WHERE clauses. The second set (marked as $bulk$) consists of those queries that are intended to crunch huge amounts of data. While the former set's input is generally striped across a small number of machines, the latter set of queries generally runs on data stored on a large number of machines,  incurring a higher communication cost. Figure~\ref{fig:scaleup} plots the query latency for each of these workloads  as a function of cluster size. Each query operates on $100n$ GB of data (where $n$ is the cluster size). So for a $10$ node cluster, each query operates on $1$ TB of data and for a $100$ node cluster each query  operates on around $10$ TB of  data. Further, for each workload suite, we evaluate the query latency for the case when the required samples are completely cached in RAM or when they are stored entirely on disk. Since in reality any sample will likely partially reside both on disk and in memory these results indicate the min/max latency bounds for any query.

\eat{
\begin{figure}[htbp]
\begin{center}
\includegraphics*[width=225pt]{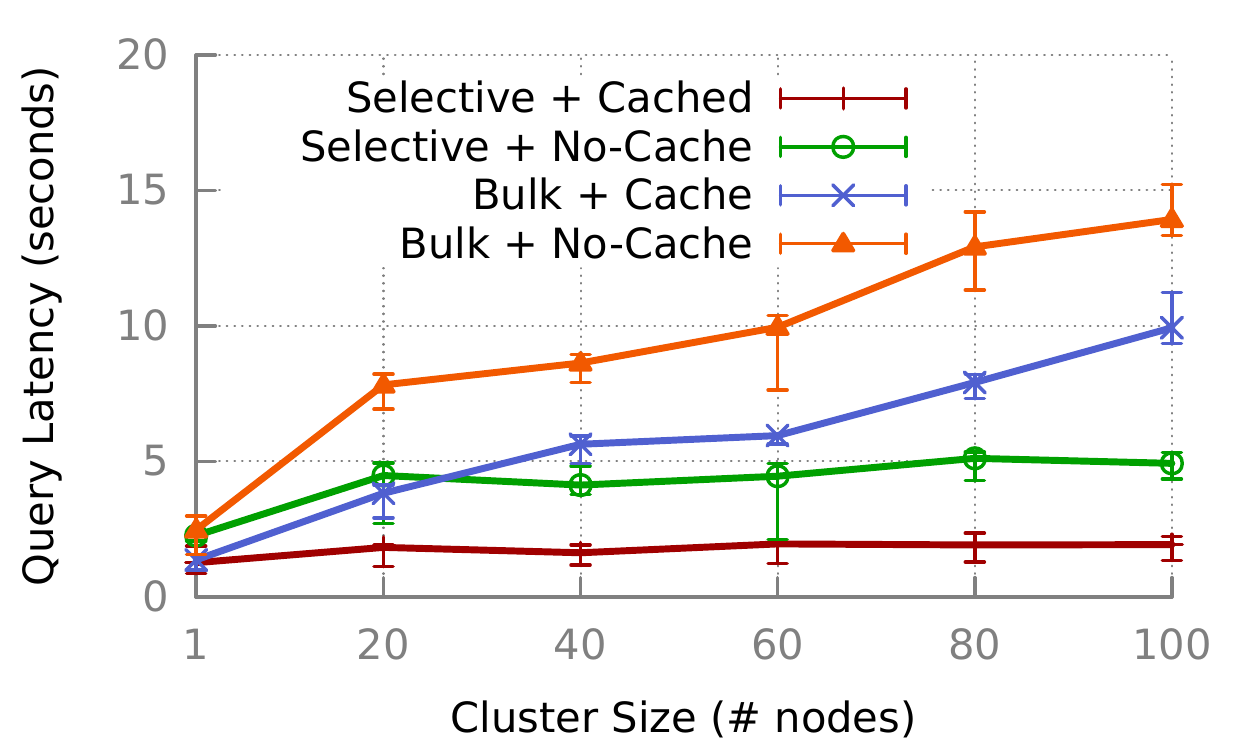}
\caption{Query latency across 2 different query workloads (with cached and non-cached samples) as a function of cluster size}
\label{fig:scaleup}
\end{center}
\end{figure}
}





\if{0}

\begin{figure}[htbp]
\begin{center}
\includegraphics*[width=225pt]{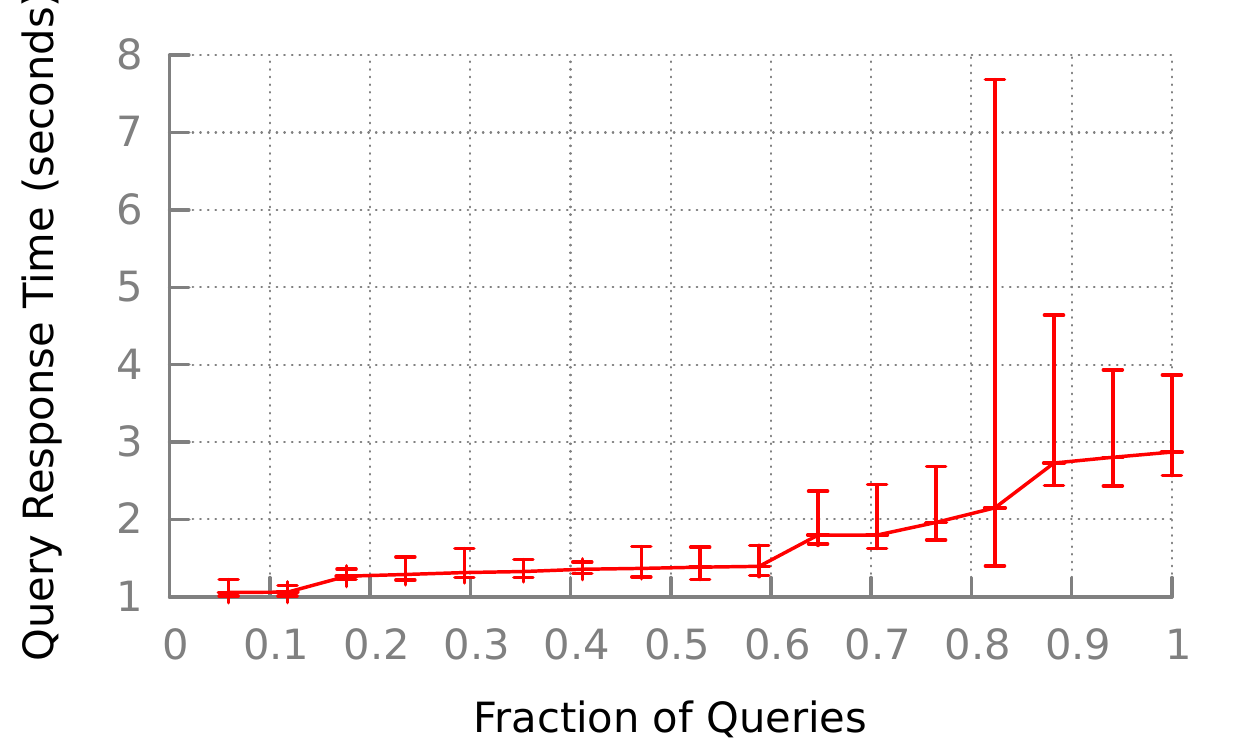}
\caption{A CDF of the response time of all queries in the work-suite ) took 20 queries ran them 10x each, with the target of 5\% error connecting the means \notesameer{Optional}}
\label{fig:query-response-time-cdf}
\end{center}
\end{figure}

\begin{figure}[htbp]
\begin{center}
\includegraphics*[width=225pt]{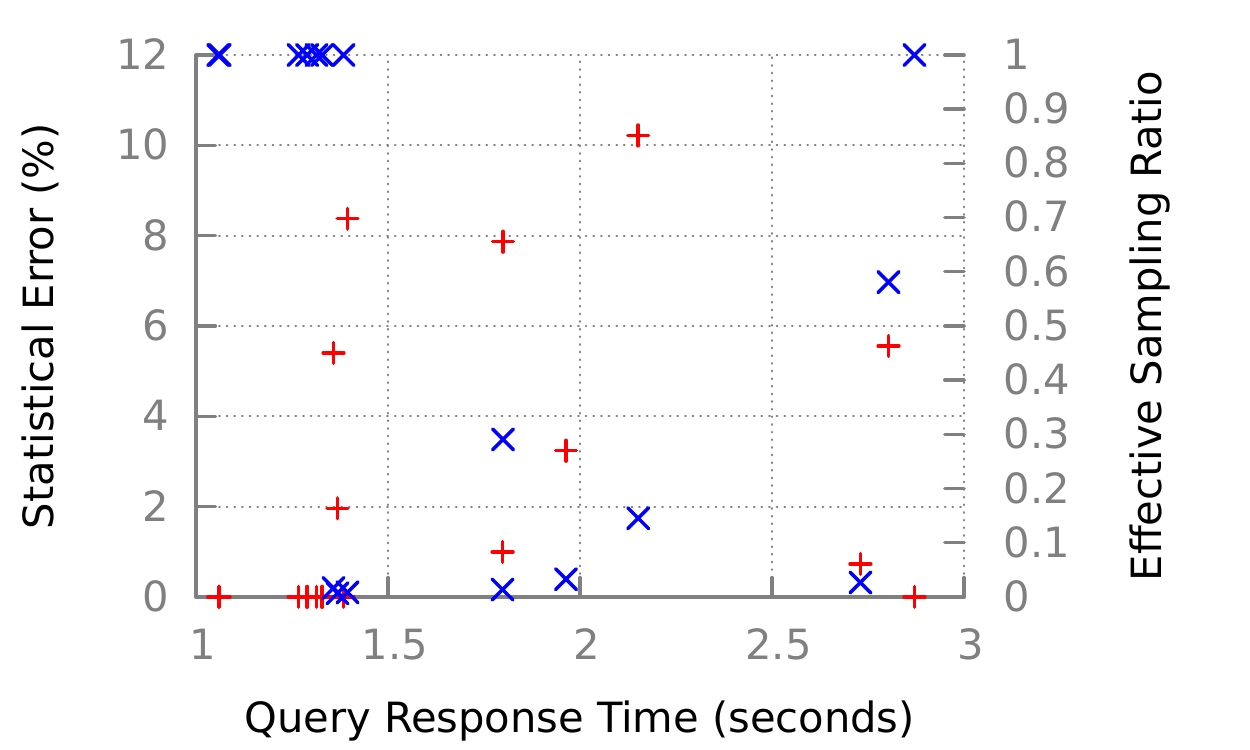}
\caption{(\notesameer{Optional}) This graph plots the statistical error ($+$) and effective sampling ratio ($X$) with respect to the query response times. This figure depicts our entire optimization space.}
\label{fig:statisticalerror-vs-samplingratio}
\end{center}
\end{figure}
\fi

\section{Related Work}
\label{related}
Prior work on interactive parallel query processing frameworks has broadly relied on two different sets
of ideas.

One set of related work has focused on using additional resources (\ie memory or CPU) 
to decrease query processing time. Examples include \emph{Spark}~\cite{spark},
\emph{Dremel}~\cite{dremel} and \emph{Shark}~\cite{shark}. While these systems deliver low-latency
response times when each node has to process a relatively small amount of data (\eg when the data can fit in the 
aggregate memory of the cluster), they become slower as the data grows unless new resources are constantly
being added in proportion.
Additionally, a significant portion of query execution time in these systems involves shuffling
or repartitioning massive amounts of data over the network, which is often
a bottleneck for queries.
By using samples, \system{} is able to scale better as the quantity of data grows.  Additionally,
being built on Spark, {\system} is able to effectively leverage the benefits provided by these
systems while using limited resources. 

Another line of work has focused on providing approximate answers with low latency, particularly in 
database systems. Approximate Query Processing (AQP) for decision support in relational
databases has been the subject of extensive research, and can either use samples, or other
non-sampling based approaches, which we describe below.

\textbf{Sampling Approaches.}  There has been substantial work on using
sampling to provide approximate responses, including work on
 stratified sampling techniques similar to
ours (see~\cite{aqp-survey} for an overview).  
Especially relevant are:

\begin{asparaenum}
\item \emph{STRAT}~\cite{surajit-optimized-stratified} relies on a single
stratified sample, chosen based on
the exact tuples accessed by each
query. In contrast \system~ uses a set of samples computed using query
templates, and is thus more amenable to ad-hoc
queries. 
\item \emph{SciBORQ}~\cite{sciborq} is a data-analytics framework designed for scientific
workloads, which uses special structures, called
\emph{impressions}. 
Impressions are biased samples where tuples are picked based on
past query results. SciBORQ targets exploratory
scientific analysis.  In contrast to \system,
SciBORQ only supports time-based constraints.
SciBORQ also does not provide any guarantees on the error margin.
\item {\it Babcock et al.}~\cite{babcock-dynamic} also describe a 
stratified sampling technique
where  biased samples are built on a single column,
in contrast to our multi-column approach. In their approach, 
queries are executed on all biased
samples whose biased column is present in the query and the
union of results is returned as the final answer.
Instead, \system~runs on a single sample, chosen based on
the current query.
\end{asparaenum}

\textbf{Online Aggregation.} Online Aggregation
(OLA)~\cite{online-agg} and its successors~\cite{online-agg-mr, ola-mr-pansare}
proposed the idea of providing approximate answers which are constantly
refined during query execution.  It provides users with an interface to stop execution
once a sufficiently good answer is found. The main disadvantages of Online
Aggregation is that it requires data to be streamed in a random
order, which can be impractical in distributed systems. While~\cite{ola-mr-pansare} proposes
some strategies for implementing OLA on Map-Reduce, their strategy involves significant changes to
the query processor. Furthermore, {\system}, unlike OLA, can store data
clustered by a primary key, or other attributes, and take advantage of
this ordering during data processing. Additionally \system~can use knowledge
about sample sizes to better allocate cluster resources (parallelism/memory) and leverage
standard distributed query optimization techniques~\cite{join-comparison-in-mr}.
\eat{
In order to provide
statistical guarantees at all stages of the query processing, OLA has
to access the data in random order.  
This is impractical when dealing with large amounts of data.  
Moreover, due to the use of stratified sampling, under the same time or error
constraints, \system{} is more likely to operate on values from the
long tail of a skewed distribution.  
Finally, Online Aggregation requires significant changes to the
query processor and user interface, 
requiring, for example, the ability to scan tables in random
error, or to maintain and output incremental answers. 
}

\textbf{Non-Sampling Approaches.}  There has been a great deal of work
on ``synopses'' for answering specific types of queries
(e.g., wavelets, histograms, sketches, etc.)\footnote{Please see~\cite{aqp-survey} for a survey}.
Similarly materialized views and data cubes can be constructed to answer
specific queries types efficiently. While offering fast responses, these techniques 
require specialized structures to be 
built for every operator, or in some cases for every type of query
and are hence impractical when processing arbitrary queries. Furthermore, 
these techniques are orthogonal to our
work, and \system~could be  modified to use
any of these techniques for better accuracy on certain types of queries, while
resorting to samples on others.

\section{Conclusion}\label{sec:conclusion}

In this paper, we presented~\system, a parallel, sampling-based approximate query
engine that provides support for ad-hoc queries with error and
response time constraints. \system~is based on two key ideas: (i) a
multi-dimensional, multi-granularity sampling strategy that builds and
maintains a large variety of samples, and (ii) a run-time dynamic sample
selection strategy that uses smaller samples to estimate query
selectivity and choose the best samples for satisfying query
constraints. Evaluation results on real data sets and on deployments
of up to $100$ nodes demonstrate the effectiveness of \system{} at handling
a variety of queries with diverse error and time constraints, allowing
us to answer a range of queries within $2$ seconds on $17$ TB of data
with 90-98\% accuracy.


\section*{Acknowledgements}
The authors would like to thank Anand Iyer, Joe Hellerstein, Michael Franklin, Surajit Chaudhuri and Srikanth Kandula for their invaluable feedback and suggestions
 throughout this project. 
The authors are also grateful to Ali Ghodsi, Matei Zaharia, Shivaram Venkataraman, Peter Bailis and Patrick Wendell for their comments on 
an early version of this manuscript.
Sameer Agarwal and Aurojit Panda are supported by the Qualcomm Innovation Fellowship during 2012-13. This research is also supported in part by gifts from Google, SAP, Amazon Web Services, Blue Goji, Cloudera, Ericsson, General Electric, Hewlett Packard, Huawei, IBM, Intel, MarkLogic, Microsoft, NEC Labs, NetApp, Oracle, Quanta, Splunk, VMware and by DARPA (contract \#FA8650-11-C-7136).

\let\oldthebibliography=\thebibliography
\let\endoldthebibliography=\endthebibliography
\renewenvironment{thebibliography}[1]{%
    \begin{oldthebibliography}{#1}%
      \setlength{\parskip}{0ex}%
      \setlength{\itemsep}{0ex}%
}%
{%
\end{oldthebibliography}%
}

{\small
\bibliographystyle{abbrv}
\bibliography{references}  

\begin{thebibliography}{10}

\bibitem{hdfs}
{Apache Hadoop Distributed File System}.
\newblock \url{http://hadoop.apache.org/hdfs/}.

\bibitem{hadoopmr}
Apache {H}adoop {M}apreduce {P}roject.
\newblock \url{http://hadoop.apache.org/mapreduce/}.

\bibitem{Conviva}
Conviva {I}nc.
\newblock \url{http://www.conviva.com/}.

\bibitem{glpk}
G{N}{U} {L}inear {P}rogramming {K}it.
\newblock \url{www.gnu.org/s/glpk/}.

\bibitem{qubole}
{Qubole: Fast Query Authoring Tools}.
\newblock \url{http://www.qubole.com/s/overview}.

\bibitem{tpch}
{TPC-H Query Processing Benchmarks}.
\newblock \url{http://www.tpc.org/tpch/}.

\bibitem{rope}
S.~Agarwal, S.~Kandula, N.~Bruno, M.-C. Wu, I.~Stoica, and J.~Zhou.
\newblock {Re-optimizing Data Parallel Computing}.
\newblock In {\em NSDI}, 2012.

\bibitem{mantri-osdi}
G.~Ananthanarayanan, S.~Kandula, A.~G. Greenberg, I.~Stoica, Y.~Lu, B.~Saha,
  and E.~Harris.
\newblock Reining in the outliers in map-reduce clusters using mantri.
\newblock In {\em OSDI}, pages 265--278, 2010.

\bibitem{babcock-dynamic}
B.~Babcock, S.~Chaudhuri, and G.~Das.
\newblock Dynamic sample selection for approximate query processing.
\newblock In {\em VLDB}, 2003.

\bibitem{join-comparison-in-mr}
S.~Blanas, J.~M. Patel, V.~Ercegovac, J.~Rao, E.~J. Shekita, and Y.~Tian.
\newblock {A comparison of join algorithms for log processing in MapReduce}.
\newblock In {\em SIGMOD}, 2010.

\bibitem{recurring-scope}
N.~Bruno, S.~Agarwal, S.~Kandula, B.~Shi, M.-C. Wu, and J.~Zhou.
\newblock {Recurring job optimization in scope}.
\newblock In {\em SIGMOD}, 2012.

\bibitem{minitables}
M.~Cafarella, E.~Chang, A.~Fikes, A.~Halevy, W.~Hsieh, A.~Lerner, J.~Madhavan,
  and S.~Muthukrishnan.
\newblock {Data management projects at Google}.
\newblock {\em SIGMOD Rec.}, 37(1):34--38, Mar. 2008.

\bibitem{surajit-optimized-stratified}
S.~Chaudhuri, G.~Das, and V.~Narasayya.
\newblock Optimized stratified sampling for approximate query processing.
\newblock {\em TODS}, 2007.

\bibitem{Chaudhuri:1999}
S.~Chaudhuri, R.~Motwani, and V.~Narasayya.
\newblock On random sampling over joins.
\newblock In {\em SIGMOD}, 1999.

\bibitem{online-agg-mr}
T.~Condie, N.~Conway, P.~Alvaro, J.~M. Hellerstein, K.~Elmeleegy, and R.~Sears.
\newblock Mapreduce online.
\newblock In {\em NSDI}, 2010.

\bibitem{shark}
C.~Engle, A.~Lupher, R.~Xin, M.~Zaharia, M.~Franklin, S.~Shenker, and
  I.~Stoica.
\newblock Shark: Fast data analysis using coarse-grained distributed memory.
\newblock In {\em SIGMOD}, 2012.

\bibitem{aqp-survey}
M.~Garofalakis and P.~Gibbons.
\newblock Approximate query processing: Taming the terabytes.
\newblock In {\em VLDB}, 2001.
\newblock Tutorial.

\bibitem{online-agg-joins}
P.~J. Haas and J.~M. Hellerstein.
\newblock {Ripple joins for online aggregation}.
\newblock In {\em SIGMOD}, 1999.

\bibitem{control}
J.~M. Hellerstein, R.~Avnur, and V.~Raman.
\newblock {Informix under CONTROL: Online Query Processing}.
\newblock {\em {Data Min. Knowl. Discov.}}, 4(4):281--314, 2000.

\bibitem{online-agg}
J.~M. Hellerstein, P.~J. Haas, and H.~J. Wang.
\newblock Online aggregation.
\newblock In {\em SIGMOD}, 1997.

\bibitem{db-online}
C.~Jermaine, S.~Arumugam, A.~Pol, and A.~Dobra.
\newblock {Scalable approximate query processing with the DBO engine}.
\newblock In {\em SIGMOD}, 2007.

\bibitem{sampling-book}
S.~Lohr.
\newblock {\em Sampling: design and analysis}.
\newblock Thomson, 2009.

\bibitem{dremel}
S.~Melnik, A.~Gubarev, J.~J. Long, G.~Romer, S.~Shivakumar, M.~Tolton, and
  T.~Vassilakis.
\newblock Dremel: interactive analysis of web-scale datasets.
\newblock {\em Commun. ACM}, 54:114--123, June 2011.

\bibitem{ola-mr-pansare}
N.~Pansare, V.~R. Borkar, C.~Jermaine, and T.~Condie.
\newblock {Online Aggregation for Large MapReduce Jobs}.
\newblock {\em PVLDB}, 4(11):1135--1145, 2011.

\bibitem{sciborq}
L.~Sidirourgos, M.~L. Kersten, and P.~A. Boncz.
\newblock {SciBORQ: Scientific data management with Bounds On Runtime and
  Quality.}
\newblock In {\em CIDR'11}, 2011.

\bibitem{subset-error}
N.~Tatbul and S.~Zdonik.
\newblock Window-aware load shedding for aggregation queries over data streams.
\newblock In {\em VLDB}, 2006.

\bibitem{hive}
A.~Thusoo, J.~S. Sarma, N.~Jain, Z.~Shao, P.~Chakka, S.~Anthony, H.~Liu,
  P.~Wyckoff, and R.~Murthy.
\newblock {Hive: a warehousing solution over a map-reduce framework}.
\newblock {\em PVLDB}, 2(2), 2009.

\bibitem{spark}
M.~Zaharia, M.~Chowdhury, T.~Das, A.~Dave, J.~Ma, M.~McCauley, M.~J. Franklin,
  S.~Shenker, and I.~Stoica.
\newblock Resilient distributed datasets: A fault-tolerant abstraction for
  in-memory cluster computing.
\newblock In {\em NSDI}, 2012.

\bibitem{late-osdi}
M.~Zaharia, A.~Konwinski, A.~D. Joseph, R.~H. Katz, and I.~Stoica.
\newblock I{mproving MapReduce Performance in Heterogeneous Environments}.
\newblock In {\em OSDI}, 2008.

\bibitem{kai-paper}
K.~Zeng, B.~Mozafari, S.~Gao, and C.~Zaniolo.
\newblock {Uncertainty Propagation in Complex Query Networks on Data Streams}.
\newblock Technical report, UCLA, 2011.

\end{thebibliography}
}

\appendix
\section{Stratified Sampling Properties and Storage Overhead}
\label{sec:analysis}

In this section prove the two properties stated in
Section~\ref{sec:stratified-samples} and give the storage overhead for
Zipf distribution.




\vskip 0.1in
\noindent
\textbf{Performance properties.} Recall that $S(\phi, K^{opt})$
represents the smallest possible stratified sample on $\phi$ that
satisfies the error or response time constraints of query $Q$, while
$S(\phi, K')$ is the closest sample in $SFam(\phi)$ that satisfies
$Q$'s constraints. Then, we have the following results.

\begin{lemma}
  Assume an I/O-bound query $Q$ that specifies an error
  constraint. Let $r$ be the response time of $Q$ when running on the
  optimal-sized sample, $S(\phi, K^{opt})$. Then, the response time of
  $Q$ when using sample family $SFam(\phi) = \{S(\phi, K_i)\}, (0 \leq
  i < m)$ is at most $c + 1/K^{opt}$ times larger than $r$.
\end{lemma}

\begin{proof}
 Let $i$ be such that

\begin{equation}
\left\lfloor \frac{K}{c^{i+1}} \right\rfloor < K^{opt} \leq \left\lfloor \frac{K}{c^{i}} \right\rfloor.
\label{eq:query-error-1}
\end{equation}

\noindent 
Assuming that the error of $Q$ decreases monotonically with the
increase in the sample size, $S(\phi, \lfloor K/c^{i} \rfloor)$ is the
smallest sample in $SFam(\phi)$ that satisfies $Q$'s error
constraint. Furthermore, from Eq.~(\ref{eq:query-error-1}) it follows
that

\begin{equation}
\left\lfloor \frac{K}{c^{i+1}} \right\rfloor < c K^{opt}  + 1.
\label{eq:query-error-bounds}
\end{equation}

\noindent
In other words, in the worst case, $Q$ may have to use a sample whose
cap is at most $c + 1/K^{opt}$ times larger than $K^{opt}$. Let $K' =
c K^{opt} + 1$, and let $A = \{a_1, a_2, ..., a_k\}$ be the set of
values in $\phi$ selected by $Q$. By construction, both samples
$S(\phi, K')$ and $S(\phi, K^{opt})$ contain all values in the fact
table, and therefore in set $A$. Then, from the definition of the
stratified sample, it follows that the frequency of any $a_i \in A$ in
sample $S(\phi, K')$ is at most $K'/K^{opt}$ times larger than the
frequency of $a_i$ in $S(\phi, K^{opt})$. Since the tuples matching
the same value $a_i$ are clustered together in both samples, they are
accessed sequentially on the disk. Thus, the access time of all tuples
matching $a_i$ in $S(\phi, K')$ is at most $c + 1/K^{opt}$ times
larger than the access time of the same tuples in $S(\phi,
K^{opt})$. Finally, since we assume that the $Q$'s execution is
I/O-bound, it follows that $Q$'s response time is at most $c +
1/K^{opt}$ times worse than $Q$'s response time on the optimal sample,
$S(\phi, K^{opt})$.
\end{proof}



\begin{lemma}
  Assume a query, $Q$, that specifies a response time constraint, and
  let $S(\phi, K^{opt})$ be the largest stratified sample on column
  set $\phi$ that meets $Q$'s constraint. Assume standard deviation of
  $Q$ is $\sim 1/\sqrt{n}$, where $n$ is the number of tuples selected
  by $Q$ from $S(\phi, K^{opt})$. Then, the standard deviation of $Q$
  when using sample family $SFam(\phi)$ increases by at most
  $1/\sqrt{1/c - 1/K^{opt}}$ times.
\end{lemma}

\begin{proof}
Let $i$ be such that

\begin{equation}
\left\lfloor \frac{K}{c^{i}} \right\rfloor \leq K^{opt} < \left\lfloor \frac{K}{c^{i-1}} \right\rfloor.
\label{eq:query-time-1}
\end{equation}

\noindent Assuming that the response time of $Q$ decreases
monotonically with the sample size, $S(\phi, \lfloor K/c^{i} \rfloor)$
is the largest sample in $SFam(\phi)$ that satisfies $Q$'s response
time. Furthermore, from Eq.~(\ref{eq:query-time-1}) it follows
that

\begin{equation}
  \label{eq:query-time-bounds}
\left\lfloor \frac{K}{c^{i}} \right\rfloor > \frac{K^{opt}}{c} - 1.
\end{equation}

\noindent Assuming that the number of tuples selected by $Q$ is
proportional to the sample size, the standard deviation of running $Q$
on $S(\phi, \lfloor K/c^{i} \rfloor )$ increases by at most $1/\sqrt{1/{c} -
1/K^{opt})}$ times.
\end{proof}



\noindent
\textbf{\bf Storage Overhead for Zipf distribution.} We evaluate the
storage overhead of maintaining a stratified sample, $S(\phi, K)$, for
a Zipf distribution, one of the most popular heavy tail distributions
for real-world datasets. Without loss of generality, assume $F(\phi,
T, x) = M/rank(x)^{s}$, where $rank(x)$ represents the rank of $x$ in
$F(\phi, T, x)$ (\ie value $x$ with the highest frequency has rank
$1$), and $s \geq 1$. Table~\ref{tab:zipf-overhead} shows the overhead
of $S(\phi, K)$ as a percentage of the original table size for various
values of Zipf's exponent, $s$, and for various values of $K$. The
number of unique values in the original table size is $M = 10^9$. For
$s = 1.5$ the storage required by $S(\phi, K)$ is only $2.4\%$ of the
original table for $K=10^4$, $5.2\%$ for $K=10^5$, and $11.4\%$ for
$K=10^6$.

\begin{table}[htbp]
{\small
\begin{center}
  \begin{tabular}{|c|r|r|r|}
    \hline
    $s$ & $K = 10,000$ & $K = 100,000$ & $K = 1,000,000$\\ \hline\hline
      $1.0$ & $0.49$ & $0.58$ & $0.69$ \\\hline
      $1.1$ & $0.25$ & $0.35$ & $0.48$ \\\hline
      $1.2$ & $0.13$ & $0.21$ & $0.32$ \\\hline
      $1.3$ & $0.07$ & $0.13$ & $0.22$ \\\hline
      $1.4$ & $0.04$ & $0.08$ & $0.15$ \\\hline
      $1.5$ & $0.024$ & $0.052$ & $0.114$ \\\hline
      $1.6$ & $0.015$ & $0.036$ & $0.087$ \\\hline
      $1.7$ & $0.010$ & $0.026$ & $0.069$ \\\hline
      $1.8$ & $0.007$ & $0.020$ & $0.055$ \\\hline
      $1.9$ & $0.005$ & $0.015$ & $0.045$ \\\hline
      $2.0$ & $0.0038$ & $0.012$ & $0.038$ \\\hline
\end{tabular}
\end{center}
}
\caption{The storage required to maintain sample $S(\phi, K)$ as a
  fraction of the original table size. The distribution of $\phi$ is
  Zipf with exponent $s$, and the highest frequency ($M$) of $10^9$. }
\label{tab:zipf-overhead}
\end{table}

\end{document}